\pdfoutput=1 
\documentclass[10pt]{article} 
\usepackage{etoolbox} 
\newtoggle{ForWINDOWS}
\togglefalse{ForWINDOWS}
\newtoggle{UseREVTEX}
\togglefalse{UseREVTEX}
\usepackage{amsfonts,amsmath,amssymb,amsthm,bm,bbm}
\iftoggle{UseREVTEX} {} {\usepackage{cite}}
\usepackage{enumitem}
\usepackage{graphicx}
\usepackage{lineno}
\iftoggle{UseREVTEX} {} {\usepackage{myOSAmeet}}
\usepackage{relsize}
\usepackage{rotating}
\iftoggle{ForWINDOWS} {\usepackage{subcaption}} {}
\iftoggle{ForWINDOWS} {} {\usepackage{subfig}}
\iftoggle{UseREVTEX} {} {\setlength{\topmargin}{-15mm}}
\newcommand{\defeq}{\;{\raisebox{-.1\height}{$\stackrel{\raisebox{-.1\height}{\tiny\rm def}}{=}$}}\;}
\newcommand{\Iver}[1]{\left[#1\right]_{\mathsmaller{\rm Iver}}}

\def\myTheorem{Theorem}
\newtheorem{definition}{Definition}
\newtheorem{theorem}{\myTheorem}
\def\HomoPhysR{\,{\stackrel{\mathsmaller{\mathfrak{M}}\;}{\Longrightarrow}}\,}
\def\la{\langle}
\def\ra{\rangle}
\begin{document}
\title{A Solution to the Sign Problem Using a\\ Sum of Controlled Few-Fermions}
\iftoggle{UseREVTEX} {
\author{David H. Wei\,}
\email[]{david.hq.wei@gmail.com}
\affiliation{Quantica Computing, LLC, San Jose, California}
} {
\author{David H. Wei$^{\,*}$}
\vspace{0.75ex}
\address{Quantica Computing, LLC, San Jose, California}
\address{*\,Email: david.hq.wei@gmail.com}
}
\begin{abstract}
A restricted path integral method is proposed to simulate a type of quantum system or Hamiltonian called a sum of controlled few-fermions on a classical computer using Monte Carlo without a numerical sign problem. Then a universality is proven to assert that any bounded-error quantum polynomial time (BQP) algorithm can be encoded into a sum of controlled few-fermions and simulated efficiently using classical Monte Carlo. Therefore, BQP is precisely the same as the class of bounded-error probabilistic polynomial time (BPP), namely, BPP\,=\,BQP.
\end{abstract}
\iftoggle{UseREVTEX} {\maketitle} {\vspace{0.75ex}}

The ability to simulate quantum systems efficiently on a classical computer \cite{Feynman82,Grotendorst02} or a quantum machine \cite{Feynman82,Feynman85,Lloyd96} is crucially important for fundamental sciences and practical applications. Of particular importance is to simulate the ground state of a quantum many-body system efficiently on a classical computer. Among many numerical methods, quantum Monte Carlo \cite{Grotendorst02} is uniquely advantageous as being based on first principles without uncontrolled systematic errors and using polynomially efficient importance sampling from an exponentially large Hilbert space, until its polynomial efficiency is spoiled by the notorious numerical sign problem \cite{Loh90,Ceperley91}. On the other hand, many computational problems that are not directly related to simulations of quantum systems can be solved by running a quantum algorithm on a quantum computer, which reduces to simulating a quantum system, especially the ground state of a quantum many-body system. A fundamental question in computational complexity theory is whether the class of bounded-error probabilistic polynomial time (BPP) is the same as that of bounded-error quantum polynomial time (BQP) \cite{Bernstein97}, which will be answered affirmatively in this presentation, by firstly identifying and characterizing a type of Hamiltonian called a sum of controlled few-fermions (CFFs), whose Gibbs kernels and ground states can be simulated rigorously using path integral Monte Carlo (PIMC) on a classical computer, then demonstrating that a universal BQP algorithm can be encoded into a sum of CFFs and PIMC-simulated efficiently with a mixing time \cite{Levin08} that is provably polynomial-bounded.

Let a triple $({\cal C},{\cal H},{\cal B})$ represent a general quantum system \cite{Wei20}, where ${\cal C}$ is a {\it configuration space} consisting of {\it configuration points}, each of which is a tuple or vector of {\it variable values} ({\it e.g.}, eigenvalues) assigned to an ensemble of {\it coordinate variables} that are dynamical variables associated with the quantum system, for example but not limited to spatial positions of particles, while ${\cal H} \defeq {\cal H}({\cal C}) \subseteq L^2({\cal C})$ is a Hilbert space of state vectors ({\it i.e.}, wavefunctions) supported by ${\cal C}$, and ${\cal B} \defeq {\cal B}({\cal H})$ is a Banach algebra of bounded operators acting on vectors in ${\cal H}$, which contains a strongly continuous one-parameter semigroup of Gibbs operators $\{\exp(-\tau H)\}_{\tau\in[0,\infty)}$, whose infinitesimal generator $H$ is designated as the Hamiltonian governing the quantum system. A Hamiltonian is the epitome of a general lower-bounded self-adjoint operator $h$, called a {\it partial Hamiltonian}, whose eigenvalues are denoted by $\{\lambda_n(h)\}_{n=0}^{\infty} \subseteq \mathbb{R}$ in a strictly increasing order, and the associated eigenstates are denoted by $\{\psi_n(h)\}_{n=0}^{\infty}$, each of which is either a non-degenerate wavefunction or a suitable representation of an eigen subspace. A partial Hamiltonian $h$ and its associated Gibbs operators $\{\exp(-\tau h)\}_{\tau\in[0,\infty)}$ are also said to be supported by ${\cal C}$ for the sake of brevity.

It is almost always the case, here taken as an axiomatic premise, that all partial Hamiltonians in consideration are of the Schr\"odinger type, as a sum of an elliptic differential operator $-\Delta$ and a bounded potential $V$, so to substantiate the {\it Hopf lemma} and the {\it strong Hopf extremum principle} \cite{Hopf52,Evans10,Wei20}. It is without loss of generality (WLOG) to assume that ${\cal C}$ is a compact Riemannian manifold of a finite dimension $\dim({\cal C}) \in \mathbb{N}$, and all wavefunctions are real-valued, so the Hilbert spaces and Banach algebras are over $\mathbb{R}$ \cite{Bernstein97,Rudolph02,Biamonte08,Wei20}. The Riemannian metric $g$ defining ${\cal C}$ induces a length measure for curves and a volume measure $V_g$ for subsets on ${\cal C}$. Being compact, ${\cal C}$ has a finite diameter ${\rm diam}({\cal C}) \in \mathbb{R}$, thus a finite ${\rm size}({\cal C}) \defeq \dim({\cal C}) + {\rm diam}({\cal C})$. It is assumed that all fundamental equations of physics, especially the Schr\"odinger and quantum field-theoretic equations, are nondimensionalized and written in a so-called natural unit system.

Consider a {\it many-fermion system} (MFS) of a variable size, comprising a number $S\in\mathbb{N}$ of fermion species, each species labeled by $s\in[1,S]$ consisting of a number $n_s\in\mathbb{N}$ of identical fermions moving on a low-dimensional Riemannian manifold ${\cal X}_s$, where the total number of particles $N_* \defeq \sum_{s=1}^{\mathsmaller{S}}n_s$ may go up unbounded, while both $d_s \defeq \dim({\cal X}_s)$ and $D_s \defeq {\rm diam}({\cal X}_s)$, $\forall s\in[1,S]$ are always bounded by a fixed number. Mathematically, the identical fermions of each species $s\in[1,S]$ may be artificially labeled by an integer $n\in[1,n_s]$, so that a configuration coordinate (point) $q_s \defeq (q_{s1},q_{s2},\cdots\!,q_{sn_s}) \in {\cal X}_s^{n_s}$ can represent their spatial configuration, where $\forall(s,d)\in[1,S]\times[1,n_s]$, $q_{sn} \defeq (x_{sn1},\cdots\!,x_{snd_s})$ is a $d_s$-dimensional coordinate of the $n$-th fermion of the $s$-th species, and $\forall d\in[1,d_s]$, $x_{snd}$ is a coordinate along the $d$-th dimension and counted as one {\it degree of freedom}. But physically, the indistinguishability among identical fermions dictates that all of the label-exchanged coordinates be equivalent and form an orbit ${\cal G}_sq_s \defeq \{\pi_sq_s : \pi_s \in {\cal G}_s\}$ for any $q_s \in {\cal X}_s^{n_s}$, where ${\cal G}_s$ is the symmetry group of permuting $n_s$ labels, $\pi_s \in {\cal G}_s$ is a typical permutation. Straightforwardly, the Cartesian product ${\cal C} \defeq \prod_{s=1}^{\mathsmaller{S}} {\cal X}_s^{n_s}$ is a configuration space for the MFS, and the group direct product $G_* \defeq \prod_{s=1}^{\mathsmaller{S}} {\cal G}_s$, called the {\it exchange symmetry group} of the MFS, acts on ${\cal C}$ and partitions it into disjoint orbits. Clearly, every pair of permutations $\pi \in {\cal G}_s$, $s\in[1,S]$ and $\pi' \in {\cal G}_{s'}$, $s'\in[1,S]$ with $s\neq s'$ commute, hence each ${\cal G}_s$, $s\in[1,S]$ is straightforwardly a normal subgroup of $G_*$. All of the even permutations in $G_*$ form a subgroup $A_*$, called the {\it exchange alternating group}. It is an axiom of physics that any legitimate quantum state $\psi \in {\cal H}({\cal C})$ must be exchange-symmetric as $[\pi\psi](q) \defeq \psi(\pi q) = (-1)^{\pi} \psi(q)$, $\forall q \in {\cal C}$, $\forall \pi \in G_*$. With respect to $G_*$ and its actions on ${\cal C}$ and ${\cal H}({\cal C})$, a {\it full exchange symmetrization operator} is defined as ${\cal F} \defeq |G_*|^{-1} \sum_{\,\pi\in G_*} (-1)^{\pi\,} \pi$, with $|G_*|$ denoting the cardinality of $G_*$ as a set. Similarly, an {\it even exchange symmetrization operator} is defined as ${\cal E} \defeq |A_*|^{-1} \sum_{\,\pi\in A_*} (-1)^{\pi\,} \pi$.

For any partial Hamiltonian $h$, and any $(r,q,\tau) \in {\cal C}^2 \times (0,\infty)$, let $\la r|e^{-\tau h}|q\ra$ represent an artificial, non-negative definite, boltzmannonic Gibbs transition amplitude from $q$ to $r$ in (imaginary) time $\tau$ due to $h$, which ignores the fermionic exchange symmetry and regards all particles distinguishable, let
\begin{align}
\la r|e^{-\tau h}|{\cal F}q\ra \,\defeq\,& |G_*|^{-1} \, {\textstyle{ \scalebox{1.25}{$\sum$}_{\pi\in G_*} }} (-1)^{\pi} \la r|e^{-\tau h}|\pi q\ra \,, \\[0.75ex]
\la{\cal F}r|e^{-\tau h}|q\ra \,\defeq\,& |G_*|^{-1} \, {\textstyle{ \scalebox{1.25}{$\sum$}_{\pi\in G_*} }} (-1)^{\pi} \la\pi r|e^{-\tau h}|q\ra \,, \\[0.75ex]
\la{\cal F}r|e^{-\tau h}|{\cal F}q\ra \,\defeq\,& |G_*|^{-2} \, {\textstyle{ \scalebox{1.25}{$\sum$}_{\pi\in G_*} \, \scalebox{1.25}{$\sum$}_{\pi'\in G_*} }} (-1)^{\pi+\pi'} \la\pi r|e^{-\tau h}|\pi'q\ra
\end{align}
denote a pre-, post-, and dual-symmetrized fermionic Gibbs transition amplitude. The function $\la\cdot|e^{-\tau h}|\cdot\ra \in L^2({\cal C}^2)$ is called the {\it boltzmannonic Gibbs kernel}, and $\la{\cdot}|e^{-\tau h}|{\cal F}\cdot\ra$, $\la{\cal F}{\cdot}|e^{-\tau h}|\cdot\ra$, $\la{\cal F}{\cdot}|e^{-\tau h}|{\cal F}\cdot\ra \in L^2({\cal C}^2)$ are called the pre-, post-, dual-symmetrized fermionic Gibbs kernels respectively, being associated with a formal Gibbs operator $\exp(-\tau h)$, $\tau\in(0,\infty)$ generated by a partial Hamiltonian $h$. It is obvious that $\la{\cdot}|e^{-\tau h}|{\cal F}\cdot\ra \equiv \la{\cal F}{\cdot}|e^{-\tau h}|\cdot\ra \equiv \la{\cal F}{\cdot}|e^{-\tau h}|{\cal F}{\cdot}\ra$, either one may be referred to as the {\it fermionic Gibbs kernel}.

{\it All Feynman path integrals, Gibbs transition amplitudes, and Gibbs kernels in this presentation should be interpreted in the boltzmannonic sense as non-negative definite quantities, unless a full exchange symmetrization operator ${\cal F}$ is placed explicitly in front of one or more configuration coordinate(s) to signify full exchange symmetrization and summation of signed contributions to yield a fermionic quantity.}

Given an MFS governed by a Hamiltonian $H$, a computationally important number is the (descriptive) size of $H$, denoted by ${\rm size}(H)$, which is basically, up to a constant factor, the minimum number of classical bits needed to describe $H$. All computational complexities and singular values of operators will be measured against ${\rm size}(H)$. Of great interests are the so-called (computationally) local Hamiltonians \cite{Kitaev02,Kempe06} of the form $H=\sum_{k=1}^{\mathsmaller{K}}H_k$, $K\in\mathbb{N}$, $K=O({\rm poly}(N_*))$, where each $H_k$, $k \in [1,K]$ moves no more than a constant number of degrees of freedom around any configuration point, the size of $H$ is defined as ${\rm size}(H) \defeq {\rm size}({\cal C}) + K$, so ${\rm size}(H) = O({\rm poly}(N_*))$. In this presentation, $O(\cdot)$, $\Omega(\cdot)$, and $\Theta(\cdot)$ are the traditional notations of asymptotics in the Knuth convention, representing an upper bound, a lower bound, and a simultaneous upper and lower bound, respectively \cite{Knuth76}. An $H_k$, $k \in [1,K]$ is said to move an $(s,n,d)$-th degree of freedom, $(s,n,d)\in[1,S]\times[1,n_s]\times[1,d_s]$ around a configuration point $q = (\cdots\!,q_{snd},\cdots) \in {\cal C}$, when there exist a $\tau \in (0,\infty)$ and an $r = (\cdots\!,r_{snd},\cdots) \in {\cal C}$, such that $r_{snd} \neq q_{snd}$ while the boltzmannonic Gibbs transition amplitude $\la r|e^{-\tau H_k}|q\ra \neq 0$ \cite{Wei20}. For any $k\in[1,K]$ and any $q\in{\cal C}$, let ${\cal R}_k(q) \defeq \{r\in{\cal C} : \exists \, \tau > 0\mbox{ such that }\la r|e^{-\tau H_k}|q\ra \neq 0\}$, which is called the {\it boltzmannonic reach} of $q$ by $H_k$.

One exemplary local Hamiltonian describes a type of MFS called a {\it few-species fermionic system} (FSFS), which comprises a small number $S\in\mathbb{N}$ of different fermion species, one or more species having a large number of identical particles, where the particles are artificially labeled so that a conventional Schr\"odinger operator of the FSFS is rewritten deliberately into a local Hamiltonian $H = \sum_{s=1}^{\mathsmaller{S}} \sum_{l=1}^{n_s} \sum_{m=1}^{n_s} H_{slm}$, where for each $(s,l,m)\in[1,S]\times[1,n_s]^2$, $H_{slm} \defeq {-}(\Delta_{sl}+\Delta_{sm})/2n_s + V/\sum_{s=1}^{\mathsmaller{S}}n_s^2$ moves only the $l$-th and the $m$-th labeled particle of the $s$-th species through the Laplace-Beltrami operators $\Delta_{sl}$ and $\Delta_{sm}$. A {\it single-species fermionic system} is a special case of FSFS with $S=1$, namely, involving a single fermion species having a large number of identical particles being artificially labeled.

Another exemplary local Hamiltonian $H=\sum_{k=1}^{\mathsmaller{K}}H_k$ describes an important type of MFS called a {\it many-species fermionic system} (MSFS), which comprises a large number $S\in\mathbb{N}$ of fermion species, each of which has no more than a small constant of identical fermions, where each $H_k$, $k\in[1,K]$ moves particles of no more than a small constant number of species around any given configuration point $q\in{\cal C}$.

\begin{definition}{} \label{defiSCFFHamil}
Given an MSFS with a configuration space ${\cal C}$ coordinating a large number of fermions of a variable number $S\in\mathbb{N}$ of species, a form sum $H=\sum_{k=1}^{\mathsmaller{K}}H_k$, $K = O({\rm poly}(S))$ defining a local Hamiltonian is called a sum of CFFs, when each $H_k$, $k\in[1,K]$, called a CFF interaction, is invariant under any exchange of identical particles, namely, $\pi^{-1}H_k\pi = H_k$, $\forall \pi \in G_*$, and satisfies $\la r|e^{-\tau H_k}|q\ra = 0$ for any $r\in{\cal C}$ and $q\in{\cal C}$ that differ in more than a small constant number of degrees of freedom. Such an MSFS or Hamiltonian $H$ is said to be sum-of-CFFs (SCFF), with SCFF serving as an adjective.
\iftoggle{UseREVTEX}{}{\vspace{-1.5ex}}
\end{definition}

For all $q\in{\cal C}$ and any $k\in[1,K]$, the CFF interaction $H_k$, by its invariance under any exchange of identical particles, moves either all or none of the particles of each species around $q$. The species and fermions being moved by $H_k$, $k\in[1,K]$ around $q\in{\cal C}$ constitute a subsystem called a controlled few-fermion (CFF), which is associated with a factor subspace ${\cal C}_k(q)$ and a factor subgroup $G_k(q) \le G_*$, in the sense that, ${\cal C}'_k(q)$ and $G'_k(q) \le G_*$ exist such that ${\cal C}_k(q)\times{\cal C}'_k(q) \simeq {\cal C}$ and $G_k(q)\times G'_k(q) \simeq G_*$, with ${\cal C}_k(q)$ and $G_k(q)$ respectively coordinating and permuting the particles in said CFF. Clearly, each $G_k(q)$ is a normal subgroup of $G_*$ and itself a direct product of a small number of normal subgroups from the list $\{{\cal G}_s\}_{s\in[1,{\mathsmaller{S}}]}$. Let $A_k(q)$ denote the associated exchange alternating group, and define the corresponding full and even exchange symmetrization operators as ${\cal F}_k(q) \defeq |G_k(q)|^{-1} \sum_{\,\pi\in G_k(q)\,} (-1)^{\pi\,} \pi$ and ${\cal E}_k(q) \defeq |A_k(q)|^{-1} \sum_{\,\pi\in A_k(q)\,} (-1)^{\pi\,} \pi$.

Despite an SCFF system having a large total number of particles $N_*$, it is always computationally easy to compute the boltzmannonic and fermionic Gibbs kernels $\la\cdot|e^{-\tau H_k}|q\ra$ and $\la\cdot|e^{-\tau H_k}|{\cal F}q\ra$, either analytically or numerically, with the complexity bounded by a constant, $\forall(q,\tau)\in{\cal C}\times(0,\infty)$, $\forall k\in[1,K]$, since $\dim({\cal C}_k(q))$ is always upper-bounded by a small constant. In particular, $\forall(r,q)\in{\cal C}^2$, in the formula $\la{\cal F}r|e^{-\tau H_k}|q\ra = \sum_{\pi\in G_*} (-1)^{\pi}\la\pi r|e^{-\tau H_k}|q\ra = \la r|e^{-\tau H_k}|{\cal F}q\ra = \sum_{\pi\in G_*} (-1)^{\pi}\la r|e^{-\tau H_k}|\pi q\ra$, even though the whole group $G_*$ is used for the domain of exchange symmetry to simplify the mathematical notation, it is really only those permutations in the much smaller subgroup $G_k(q)$ or $G_k(r)$ that are active and relevant, since $\la\pi r|e^{-\tau H_k}|q\ra=0$ for all $\pi \not\in G_k(q)$ or $\la r|e^{-\tau H_k}|\pi q\ra=0$ for all $\pi \not\in G_k(r)$. Such efficient computability of $\la{\cal F}\cdot|e^{-\tau H_k}|\cdot\ra = \la\cdot|e^{-\tau H_k}|{\cal F}\cdot\ra$ for all $k\in[1,K]$ is the key for efficient simulation of an SCFF system, and by its universality, of any quantum system on a classical computer.

\begin{definition}{} \label{defiPLTKandGSP}
Let $H=\sum_{k=1}^{\mathsmaller{K}}H_k$ be a form sum defining an SCFF Hamiltonian, where $\lambda_1(H_k)-\lambda_0(H_k) = \Omega(1/{\rm poly}({\rm size}(H)))$, $\forall k\in[1,K]$. The form sum is called a Lie-Trotter-Kato (LTK) decomposition, and $H$ is called LTK-decomposed, when $\forall \epsilon > 0$, there exists an $m \in \mathbb{N}$, $m = O({\rm poly}({\rm size}(H)\,{+}\,\epsilon^{-1}))$, such that the absolute value of $\la r|\scalebox{1.1}{$\{$}\prod_{k=1}^{\mathsmaller{K}}e^{-H_k/m}\scalebox{1.1}{$\}$}^m|{\cal F}q\ra - \la r|e^{-H}|{\cal F}q\ra$ is less than $\epsilon \la r|e^{-H}|q\ra$, $\forall (r,q) \in {\cal C}^2$. The same form sum is called a ground-state projection (GSP) decomposition and $H_0\defeq H$ is called GSP-decomposed, when $\forall \epsilon > 0$, there exist an $m \in \mathbb{N}$, $m = O({\rm poly}({\rm size}(H)\,{+}\,\epsilon^{-1}))$ and a constant $A_m>0$ depending only on $m$, such that $\|A_m\scalebox{1.1}{$\{$}\prod^{\mathsmaller{K}}_{k=1\!}\mathit{\Pi}_k\scalebox{1.1}{$\}$}^m\,{-}\,\mathit{\Pi}_0\| < \epsilon$, where $\|{\cdot}\|$ denotes the operator norm, $\mathit{\Pi}_k \defeq \lim_{\,\tau\rightarrow\infty}e^{-\tau[H_k-\lambda_0(H_k)]}$ are projections to the ground state subspaces of $H_k$, $\forall k \in [0,K]$.
\iftoggle{UseREVTEX}{}{\vspace{-1.5ex}}
\end{definition}

The definition of an LTK- or GSP-decomposed Hamiltonian is inspired by the LTK product formula $e^{-\tau H\!} = \lim_{\,m\rightarrow\infty\!} \scalebox{1.1}{$\{$}\prod_{k=1}^{\mathsmaller{K}}e^{-\tau H_k/m}\scalebox{1.1}{$\}$}^m$, $\tau\in(0,\infty)$ in a suitable operator topology, which suggests to divide $[0,\tau]$ into time intervals delimited by time instants $\{\tau_n \defeq n\,\delta\tau/K\}_{n\in[0,\,N]}$, $\delta\tau \defeq \tau/m$, $N \defeq mK$, and break the Gibbs operator $e^{-\tau H\!}$ down into a sequence of Gibbs operators $\{{\sf G}_n \defeq e^{-\delta\tau H_{n\|\mathsmaller{K}}}\}_{n\in[1,\,N]}$, so to compute the boltzmannonic and fermionic Gibbs kernels using the Feynman path integral, also known as the functional integration \cite{Feynman65,Simon79}. For all $n\in\mathbb{Z}$ and any $K\in\mathbb{N}$, the expression $n\,\|\,K$ denotes the unique number such that $(n\,\|\,K) \in [1,K]$ and $(n\,\|\,K) \equiv n \! \pmod K$. Each Gibbs operator ${\sf G}_n$ and the spacetime domain ${\cal C}\times[\tau_{n-1},\tau_n]$, $n\in[1,N]$ constitute a {\it Feynman slab}, delimited by two {\it Feynman planes} $({\cal C},\tau_{n-1}) \defeq \{(q_{n-1},\tau_{n-1}) : q_{n-1} \in {\cal C}\}$ and $({\cal C},\tau_n) \defeq \{(q_n,\tau_n) : q_n \in {\cal C}\}$ \cite{Wei20}. If necessary, each Feynman slab associated with a constant CFF interaction $H_{n\|\mathsmaller{K}}$, $n\in[1,N]$ can be further divided into thinner {\it Feynman slices}, each of which is defined by two Feynman planes separated by an interval of imaginary time that is as small as desired.

For Feynman slabs or slices that are sufficiently thin, there are simple rules for {\it Feynman flights} \cite{Wei20}, which determine the associated Gibbs transition amplitude between two points $q$ and $r$ on two narrowly separated Feynman planes respectively. Said rules for Feynman flights induce a Wiener measure that assigns a non-negative Wiener density $W(\gamma) = e^{-U(\gamma)}$ to each Feynman path $\gamma$, where $U(\cdot)$ is an action functional that is linear with respect to path concatenation, namely, $U(\gamma) = U(\gamma_1) + U(\gamma_2)$ holds when two segments of Feynman paths $\gamma_1$ and $\gamma_2$ concatenate into a continuous Feynman path $\gamma \defeq \gamma_2 * \gamma_1$, which starts at the start point of a first segment $\gamma_1$, goes to the end point of $\gamma_1$ that coincides with the start point of a second segment $\gamma_2$, and continues till the end point of $\gamma_2$.

A number of consecutive Feynman slabs with the corresponding sequence of Gibbs operators $\{{\sf G}_n\}_{n\in[n_1,\,n_2]}$, $0<n_1 \le n_2\le N$ constitute a {\it Feynman stack} with the two Feynman planes $({\cal C},\tau_{n_0})$ and $({\cal C},\tau_{n_2})$ forming its boundaries \cite{Wei20}, where $n_0 \defeq n_1\,{-}\,1$, $\forall n_1 \in \mathbb{N}$. Pick two points $(q_{n_0},\tau_{n_0})$ and $(q_{n_2},\tau_{n_2})$ on the two boundary Feynman planes, the set of all Feynman paths
\begin{equation}
\Gamma(q_{n_2},\tau_{n_2};q_{n_0},\tau_{n_0}) \,\defeq\, \{\gamma(\tau) : \tau \in [\tau_{n_0},\tau_{n_2}] \mapsto {\cal C} \mbox{ such that } \gamma(\tau_{n_0}) = q_{n_0},\,\gamma(\tau_{n_2}) = q_{n_2}\}
\end{equation}
constitutes a {\it Feynman spindle} \cite{Wei20}, which gives rise to a boltzmannonic Gibbs transition amplitude
\begin{align}
\rho\kern0.1em(q_{n_2},\tau_{n_2};q_{n_0},\tau_{n_0}) \,\defeq\,& \scalebox{1.25}{$\int$}_{\!\gamma\in\Gamma(q_{n_2},\tau_{n_2};q_{n_0},\tau_{n_0})} \, W(\gamma) \, d\gamma \label{StackGibbsBoltz} \\[0.75ex]
\,=\,\;& \scalebox{1.25}{$\int$} \!\cdots\! \scalebox{1.25}{$\int$} \, \scalebox{1.3}{$\{$} \, \scalebox{1.15}{$\prod$}^{n_2}_{n=n_1} \la q_n|e^{-\delta\tau H_{n\|\mathsmaller{K}}}|q_{n-1}\ra \scalebox{1.3}{$\}$} \, \scalebox{1.3}{$\{$} \, \scalebox{1.15}{$\prod$}^{n_2-1}_{n=n_1} dq_n \scalebox{1.3}{$\}$} \,, \nonumber
\end{align}
which can be exchange-symmetrized to produce a fermionic Gibbs transition amplitude
\begin{equation}
{\textstyle{ \rho\kern0.1em(q_{n_2},\tau_{n_2};{\cal F}q_{n_0},\tau_{n_0}) \,\defeq\, |G_*|^{-1} \, \scalebox{1.25}{$\sum$}_{\,\pi\in G_*} (-1)^{\pi} \rho\kern0.1em(q_{n_2},\tau_{n_2};\pi q_{n_0},\tau_{n_0}) \,. }} \label{StackGibbsFermi}
\end{equation}
Note that the fermionic exchange symmetry is ignored in equation (\ref{StackGibbsBoltz}) and the integration over all particle-distinguished Feynman paths yields a non-negative definite, boltzmannonic Gibbs transition amplitude. Fig.~\ref{FigFeynmanSpindles} (left) shows a Feynman spindle from a start point $(q_{n_0},\tau_{n_0})$ to an end point $(q_{n_2},\tau_{n_2})$, with the gray ellipses indicating the presence of infinitely many other Feynman paths connecting the points $(q_{n_0},\tau_{n_0})$ and $(q_{n_2},\tau_{n_2})$, where each Feynman path is assigned a non-negative Wiener density, and the integration of such Wiener density yields the boltzmannonic $\rho\kern0.1em(q_{n_2},\tau_{n_2};q_{n_0},\tau_{n_0})$. Fig.~\ref{FigFeynmanSpindles} (right) illustrates how exchange symmetrization induces an {\it orbit of Feynman spindles} from the $G_*$-orbit of the start point $G_*q_{n_0\!} = \{\pi q_{n_0\!} : \pi \in G_*\}$ to the end point  $q_{n_2}$, where each Feynman spindle $\Gamma(q_{n_2},\tau_{n_2};\pi q_{n_0},\tau_{n_0})$, $\pi \in G_*$ gives rise to a boltzmannonic $\rho\kern0.1em(q_{n_2},\tau_{n_2};\pi q_{n_0},\tau_{n_0})$, which is weighted by $(-1)^{\pi}$ and makes a signed contribution to the fermionic $\rho\kern0.1em(q_{n_2},\tau_{n_2};{\cal F}q_{n_0},\tau_{n_0})$. The ellipsis in Fig.~\ref{FigFeynmanSpindles} (right) indicates the presence of many Feynman spindles starting from exchange-permuted configuration points.

\begin{figure}[ht]
\centering
\begin{tabular}{cccc}
\hspace{-1.5em} &
\raisebox{0.25\height} { \subfloat { \includegraphics [width=0.45\textwidth] {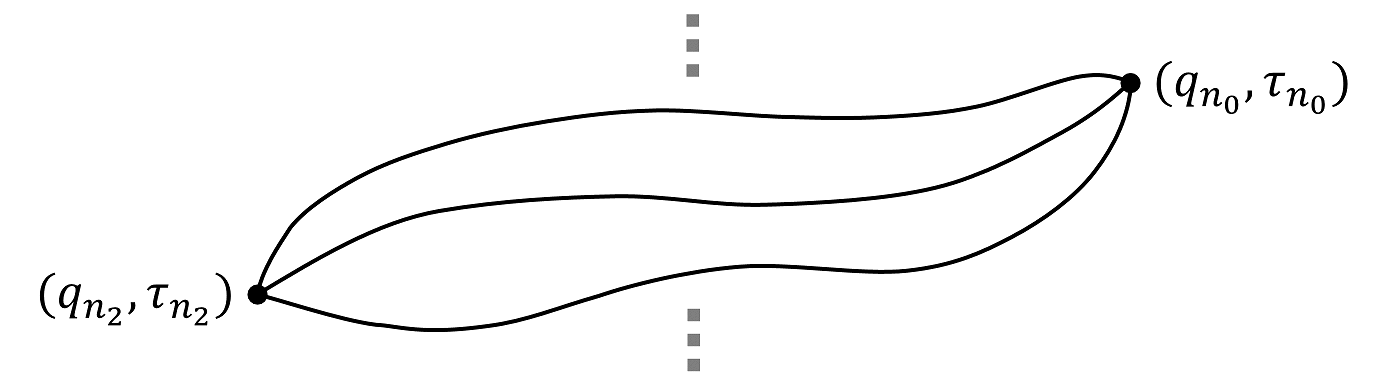} \label{fig:FeynmanSpindle} } }
& \hspace{-2.5em} &
\subfloat { \includegraphics [width=0.45\textwidth] {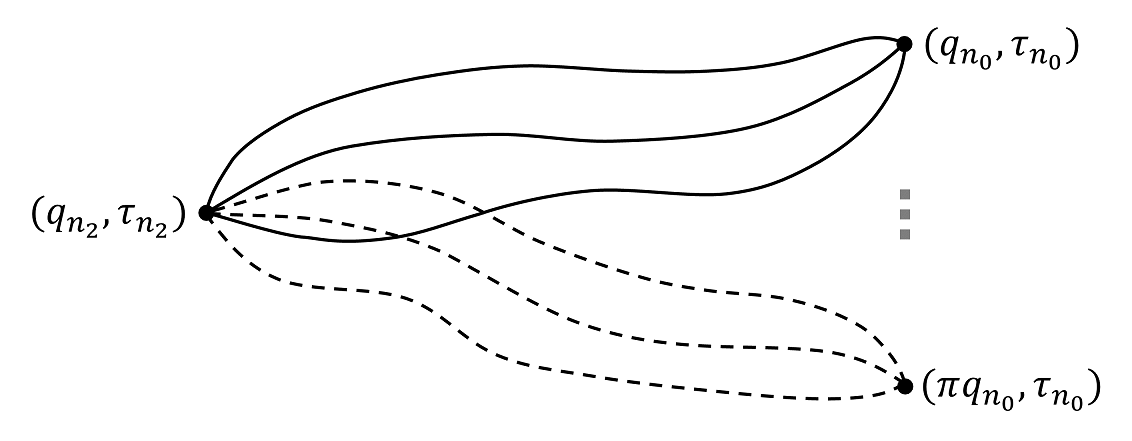} \label{fig:FeynmanSpindle2} }
\end{tabular}
\vspace{-1.5ex}
\captionsetup{font={small}}
\captionsetup{width=0.85\textwidth}
\captionsetup{justification=centering}
\caption{Left: A Feynman spindle from a start point $(q_{n_0},\tau_{n_0})$ to an end point $(q_{n_2},\tau_{n_2})$; Right: An orbit of Feynman spindles from a start point $(q_{n_0},\tau_{n_0})$ to an an orbit of end points $(G_*q_{n_2},\tau_{n_2})$.}
\label{FigFeynmanSpindles}
\end{figure}
\vspace{2.0ex}

The formulation of Feynman path integral is rightly suited for simulating a Gibbs kernel via Monte Carlo integration over a many- but finite-dimensional space. PIMC would realize BPP simulations of quantum systems, were it not for the sign problem \cite{Loh90,Ceperley91} due to the presence of negative amplitudes, particularly in fermionic systems. PIMC methods using {\it restricted path integrals} (RPIs) \cite{Ceperley91,Ceperley96,Zhang97,Zhang04,Kruger08} have been proposed and applied to avoid negative amplitudes. But previous RPIs are only approximate methods as they rely on {\it a priori} approximations for the nodal surfaces of Gibbs kernels associated with the Hamiltonian of a whole system, which are unknown and hard to compute. Here I will show that for an SCFF Hamiltonian, negative amplitudes can be avoided by restricting Feynman paths locally, with respect to the efficiently computable nodal surface of a Gibbs kernel associated with an individual CFF interaction.

For a single Feynman slab associated with a CFF interaction $H_{n_1\|\mathsmaller{K}}$ between two Feynman planes $({\cal C},\tau_{n_0})$, $({\cal C},\tau_{n_1})$, $\tau_{n_1} > \tau_{n_0}$, $n_1\in\mathbb{N}$, $n_0=n_1-1$, the pre-symmetrized fermionic Gibbs kernel $\rho\kern0.1em(q,\tau;{\cal F}q_{n_0},\tau_{n_0}) \defeq$ $\la q|e^{-(\tau-\tau_{n_0})H_{n_1\|\mathsmaller{K}}}|{\cal F}q_{n_0}\ra$ is a $(q,\tau)$-jointly continuous function of $(q,\tau) \in {\cal C} \times (\tau_{n_0},\infty)$ for any fixed $q_{n_0\!}\in{\cal C}$. The preimage $\{(q,\tau) : \rho\kern0.1em(q,\tau;{\cal F}q_{n_0},\tau_{n_0}\}$ is an open set, in which the unique connected component that contains the trivial path $\{(q_{n_0},\tau) : \tau \in (\tau_{n_0},\infty)\}$, denoted by ${\cal T}(;q_{n_0},\tau_{n_0})$, is called the {\it forward nodal tube} or {\it Ceperley reach} of $(q_{n_0},\tau_{n_0})$ \cite{Ceperley91,Wei20}. For any $\tau \in (\tau_{n_0},\infty)$, let ${\cal N}(\tau;q_{n_0},\tau_{n_0}) \defeq {\cal T}(;q_{n_0},\tau_{n_0}) \,\cap\, ({\cal C} \times \{\tau\})$, which is clearly the nodal cell of $\rho\kern0.1em(\cdot,\tau;{\cal F}q_{n_0},\tau_{n_0})$ containing the point $\cdot=q_{n_0}$. It follows from $H_{n_1\|\mathsmaller{K}}$ substantiating the Hopf lemma and the strong Hopf extremum principle that ${\cal N}(\tau;q_{n_0},\tau_{n_0})$ as an open set-valued function of $\tau\in(\tau_{n_0},\infty)$ is continuous; $\forall(q,\tau)\in{\cal C}\times(\tau_{n_0},\infty)$, $(q,\tau)\in{\cal T}(;q_{n_1},\tau_{n_1})$ if and only if $q\in{\cal N}(\tau;q_{n_0},\tau_{n_0})$ and a curve within ${\cal N}(\tau;q_{n_0},\tau_{n_0})$ exists to connect $q$ and $q_{n_0}$. Similarly, with respect to any fixed $(q_{n_1},\tau_{n_1})$ and the post-symmetrized fermionic Gibbs kernel $\rho\kern0.1em({\cal F}q_{n_1},\tau_{n_1};q,\tau)$, $(q,\tau) \in {\cal C} \times (-\infty,\tau_{n_1})$, define a {\it backward nodal tube} or {\it Ceperley reach} ${\cal T}(q_{n_1},\tau_{n_1};)$, as the unique connected component of the open set $\{(q,\tau) : \rho\kern0.1em({\cal F}q_{n_1},\tau_{n_1};q,\tau) > 0\}$ that contains the trivial path $\{(q_{n_1},\tau) : \tau \in (-\infty,\tau_{n_1})\}$, then define ${\cal N}(q_{n_1},\tau_{n_1};\tau) \defeq {\cal T}(q_{n_1},\tau_{n_1};) \,\cap\, ({\cal C} \times \{\tau\})$, $\forall \tau \in (-\infty,\tau_{n_1})$, which is clearly the nodal cell of $\rho\kern0.1em({\cal F}q_{n_1},\tau_{n_1};(\cdot,\tau)$ containing the point $\cdot=q_{n_1}$. Also similarly, ${\cal N}(q_{n_1},\tau_{n_1};\tau)$ as an open set-valued function of $\tau\in(-\infty,\tau_{n_1})$ is continuous; $\forall(q,\tau)\in{\cal C}\times(-\infty,\tau_{n_1})$, $(q,\tau)\in{\cal T}(q_{n_1},\tau_{n_1};)$ if and only if $q\in{\cal N}(q_{n_1},\tau_{n_1};\tau)$ and a curve within ${\cal N}(q_{n_1},\tau_{n_1};\tau)$ exists to connect $q$ and $q_{n_1}$.

In a {\it direct fermion path integral} method \cite{Ceperley96}, the fermionic Gibbs kernel $\{\rho\kern0.1em(q_{n_1},\tau_{n_1};{\cal F}q_{n_0},\tau_{n_0}) : (q_{n_1},q_{n_0}) \in {\cal C}^2$ may be computed using two nested loops, where an outer loop walks the configuration points $q_{n_1}\! \in {\cal C}$ and $q_{n_0\!} \in {\cal C}$, while an inner loop regards $q_{n_1}$ and $q_{n_0}$ as being fixed, repeatedly draws a random Feynman path $\gamma$ from the orbit of Feynman spindles $\Gamma(q_{n_1},\tau_{n_1};G_*q_{n_0},\tau_{n_0}) \defeq \bigcup_{\,\pi\in G_*\!}\Gamma(q_{n_1},\tau_{n_1};\pi q_{n_0},\tau_{n_0})$ and integrates the signed Wiener density $(-1)^{\pi} W(\gamma)$. The cancellation of positive and negative amplitudes severely degrades the efficacy of such direct Monte Carlo integration of a signed measure. It turns out that any Feynman path crossing or touching the boundary $\partial{\cal T}(;q_{n_0},\tau_{n_0})$ of the nodal tube ${\cal T}(;q_{n_0},\tau_{n_0})$ belongs to an orbit of {\it post-tethered Feynman spindles} whose singed amplitude contributions cancel exactly.

\vspace{-1.0ex}
\begin{figure}[ht]
\centering
\includegraphics [width=0.6\textwidth] {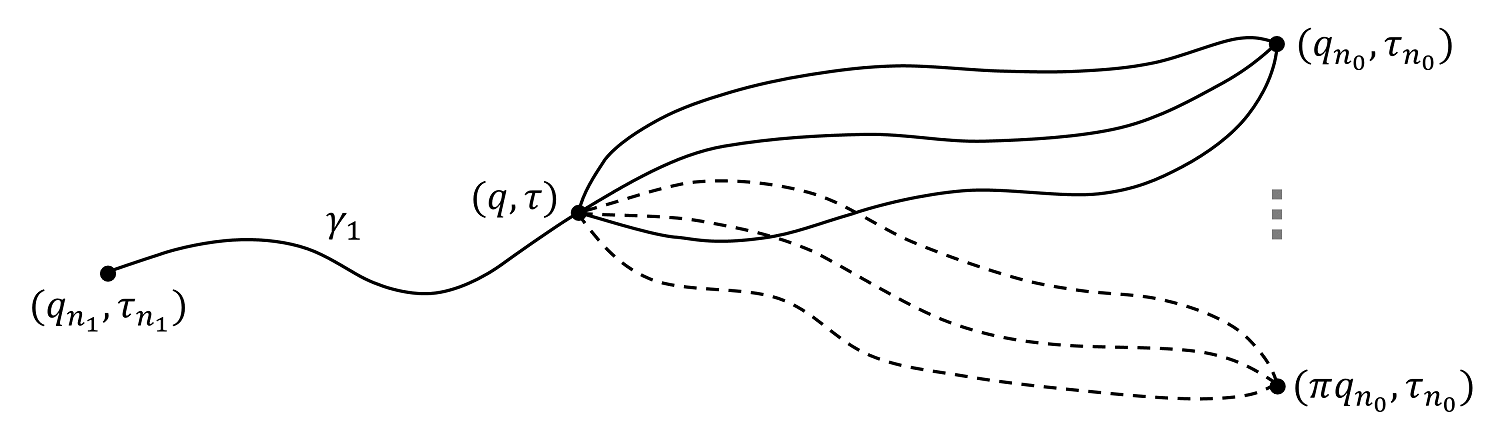}
\vspace{-1.5ex}
\captionsetup{font={small}}
\captionsetup{width=0.9\textwidth}
\captionsetup{justification=centering}
\caption{An orbit of Feynman spindles from an orbit of start points $(G_*q_{n_0},\tau_{n_0})$ to a midpoint $(q,\tau)$, that is post-tethered by a segment of Feynman path $\gamma_1$ from the midpoint to an end point $(q_{n_1},\tau_{n_1})$.}
\label{FigFeynmanSpindlePostTether}
\end{figure}
\vspace{2.0ex}

Fig.~\ref{FigFeynmanSpindlePostTether} illustrates an orbit of post-tethered Feynman spindle, which is a set of Feynman paths that come from an orbit of start points $(G_*q_{n_0},\tau_{n_0})$ to a midpoint $(q,\tau)$ via all the different ways, then share a common segment of Feynman path $\gamma_1$ from the midpoint to an end point $(q_{n_1},\tau_{n_1})$. For each $\pi \in G_*$, the set of concatenated Feynman paths $\gamma_1 * \Gamma(q,\tau;\pi q_{n_0},\tau_{n_0}) \,\defeq\, \{\gamma_1 * \gamma_0 : \gamma_0 \in \Gamma(q,\tau;\pi q_{n_0},\tau_{n_0})\}$ constitutes one {\it post-tethered Feynman spindle}, which yields a non-negative definite Wiener measure
\begin{equation}
\scalebox{1.25}{$\int$}_{\!\gamma_0\in\Gamma(q,\tau;\pi q_{n_0},\tau_{n_0})} \, W(\gamma_1) \, W(\gamma_0) \, d\gamma_0 \,=\, W(\gamma_1) \, \rho\kern0.1em(q,\tau;\pi q_{n_0},\tau_{n_0}) \,.
\end{equation}
With $\pi$ traversing the group $G_*$, or only the subgroup $G_{n_1}(q_{n_0})$ indeed, the orbit of post-tethered Feynman spindles $\gamma_1 * \Gamma(q,\tau;G_*q_{n_0},\tau_{n_0}) \defeq \{\gamma_1 * \Gamma(q,\tau;\pi q_{n_0},\tau_{n_0})\}_{\pi\in G_*}$ is enumerated, with the corresponding Wiener measures signed accordingly and summed up to yield a fermionic transition amplitude
\begin{equation}
\scalebox{1.25}{$\sum$}_{\,\pi\in G_*} \, (-1)^{\pi} \, W(\gamma_1) \, \rho\kern0.1em(q,\tau;\pi q_{n_0},\tau_{n_0}) \,=\, W(\gamma_1) * \rho\kern0.1em(q,\tau;{\cal F}q_{n_0},\tau_{n_0}) \,,
\end{equation}
which becomes exactly zero when $(q,\tau) \in \partial{\cal T}(;q_{n_0},\tau_{n_0})$. Therefore, to compute the fermionic Gibbs kernel $\rho\kern0.1em(q_{n_1},\tau_{n_1};{\cal F}q_{n_0},\tau_{n_0})$, it is sufficient to integrate over the set of ${\cal T}(;q_{n_0},\tau_{n_0})$-restricted Feynman paths
\begin{equation}
\Gamma_{\mathsmaller{\in}}(q_{n_1},\tau_{n_1};q_{n_0},\tau_{n_0}) \,\defeq\, \{\gamma(\tau) \subseteq {\cal T}(;q_{n_0},\tau_{n_0}) : \tau \in [\tau_{n_0},\tau_{n_1}], \, \gamma(\tau_{n_0}) = q_{n_0}, \, \gamma(\tau_{n_1}) = q_{n_1}\} \,,
\end{equation}
to obtain a {\it forward restricted path integral}
\begin{align}
\rho_{\mathsmaller{\in}}(q_{n_1},\tau_{n_1};q_{n_0},\tau_{n_0}) \,\defeq\,& \scalebox{1.25}{$\int$}_{\!\gamma\in\Gamma_{\mathsmaller{\in}}(q_{n_1},\tau_{n_1};q_{n_0},\tau_{n_0})} \, W(\gamma) \, d\gamma \label{RhoIn} \\[0.75ex]
\,=\,\;& \left\{ \begin{array}{rl}
\rho\kern0.1em(q_{n_1},\tau_{n_1};{\cal F}q_{n_0},\tau_{n_0}) \,, & \forall (q_{n_1},\tau_{n_1}) \in {\cal T}(;q_{n_0},\tau_{n_0}) \,, \\
0 \,, & \mbox{otherwise} \,.
\end{array} \right. \nonumber
\end{align}
$\Gamma_{\mathsmaller{\in}}(q_{n_1},\tau_{n_1};q_{n_0},\tau_{n_0})$ is called a {\it forward restricted Feynman spindle} connecting $(q_{n_0},\tau_{n_0})$ and $(q_{n_1},\tau_{n_1})$, which comprises Feynman paths that never cross or touch the boundary $\partial{\cal T}(;q_{n_0},\tau_{n_0})$.

By the identity $\rho\kern0.1em({\cal F}q_{n_1},\tau_{n_1};q_{n_0},\tau_{n_0}) = \rho\kern0.1em(q_{n_1},\tau_{n_1};{\cal F}q_{n_0},\tau_{n_0})$, $\forall (q_{n_1},q_{n_0}) \in {\cal C}^2$, a post-symmetrized fermionic Gibbs kernel $\rho\kern0.1em(;{\cal F}\cdot,\tau_{n_1}\cdot,\tau_{n_0})$ can always be computed via the equivalent pre-symmetrized $\rho\kern0.1em(\cdot,\tau_{n_1};{\cal F}\cdot,\tau_{n_0})$ using the forward restricted path integral. Alternatively, one may invoke the backward nodal cell ${\cal T}(q_{n_1},\tau_{n_1};)$ and consider an orbit of {\it pre-tethered Feynman spindles} $\Gamma(G_*q_{n_1},\tau_{n_1};q,\tau) * \gamma_0 \defeq \{\Gamma(\pi q_{n_1},\tau_{n_1};q,\tau;) * \gamma_0 :$ $\pi\in G_*\}$, $(q,\tau) \in {\cal C}\times(-\infty,\tau_{n_1})$ as depicted in Fig.~\ref{FigFeynmanSpindlePreTether}, so to establish the sufficiency of integrating over the set of ${\cal T}(q_{n_1},\tau_{n_1};)$-restricted Feynman paths
\begin{equation}
\Gamma_{\mathsmaller{\ni}}(q_{n_1},\tau_{n_1};q_{n_0},\tau_{n_0}) \,\defeq\, \{\gamma(\tau) \subseteq {\cal T}(q_{n_1},\tau_{n_1};) : \tau \in [\tau_{n_0},\tau_{n_1}], \, \gamma(\tau_{n_0}) = q_{n_0}, \, \gamma(\tau_{n_1}) = q_{n_1}\} \,,
\end{equation}
to obtain a {\it backward restricted path integral}
\begin{align}
\rho_{\mathsmaller{\ni}}(q_{n_1},\tau_{n_1};q_{n_0},\tau_{n_0}) \,\defeq\,& \scalebox{1.25}{$\int$}_{\!\gamma\in\Gamma_{\mathsmaller{\ni}}(q_{n_1},\tau_{n_1};q_{n_0},\tau_{n_0})} \, W(\gamma) \, d\gamma \label{RhoNi} \\[0.75ex]
\,=\,\;& \left\{ \begin{array}{rl}
\rho\kern0.1em({\cal F}q_{n_1},\tau_{n_1};q_{n_0},\tau_{n_0}) \,, & \forall (q_{n_0},\tau_{n_0}) \in {\cal T}(q_{n_1},\tau_{n_1};) \,, \\
0 \,, & \mbox{otherwise} \,.
\end{array} \right. \nonumber
\end{align}
$\Gamma_{\mathsmaller{\ni}}(q_{n_1},\tau_{n_1};q_{n_0},\tau_{n_0})$ is called a {\it backward restricted Feynman spindle} connecting $(q_{n_0},\tau_{n_0})$ and $(q_{n_1},\tau_{n_1})$, which comprises Feynman paths that never cross or touch the boundary $\partial{\cal T}(q_{n_1},\tau_{n_1};)$.

\vspace{-1.0ex}
\begin{figure}[ht]
\centering
\includegraphics [width=0.6\textwidth] {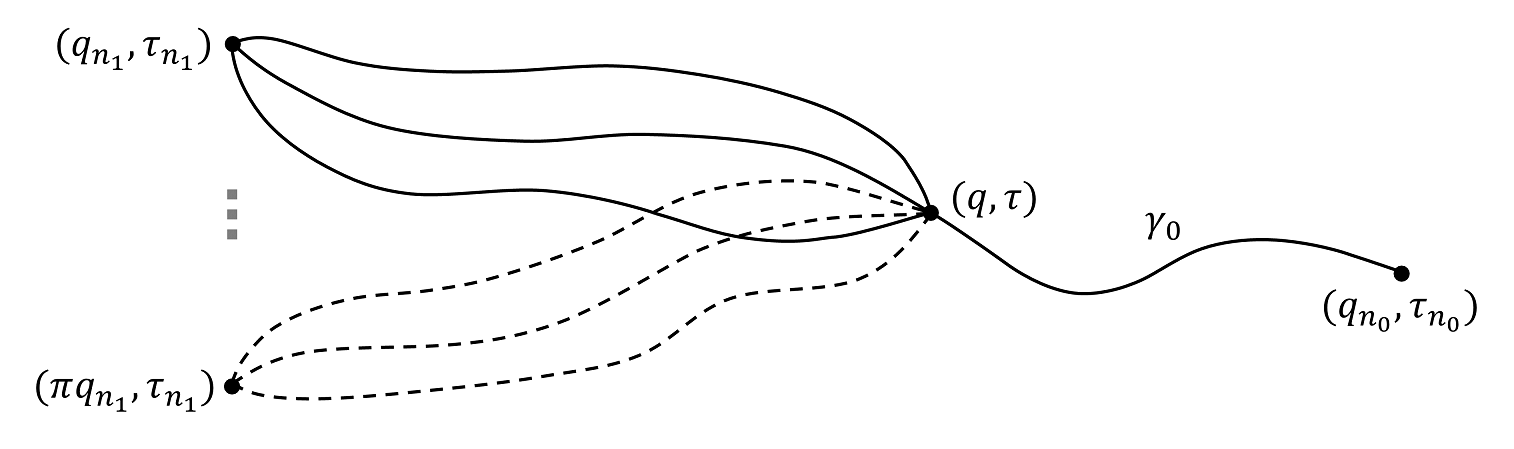}
\vspace{-1.5ex}
\captionsetup{font={small}}
\captionsetup{width=0.9\textwidth}
\captionsetup{justification=centering}
\caption{An orbit of Feynman spindles from a midpoint $(q,\tau)$ to an orbit of end points $(G_*q_{n_1},\tau_{n_1})$, that is pre-tethered by a segment of Feynman path $\gamma_0$ from a start point $(q_{n_0},\tau_{n_0})$ to the midpoint.}
\label{FigFeynmanSpindlePreTether}
\end{figure}
\vspace{2.0ex}

Therefore, for each single Feynman slab associated with a CFF interaction $H_{n_1\|\mathsmaller{K}}$, $n_1\in\mathbb{N}$, it is always easy to compute the fermionic Gibbs kernel $\rho\kern0.1em({\cal F}q_{n_1},\tau_{n_1};q_{n_0},\tau_{n_0}) = \rho\kern0.1em(q_{n_1},\tau_{n_1};{\cal F}q_{n_0},\tau_{n_0})$, $\forall(q_{n_1},q_{n_0})\in{\cal C}^2$, $\forall\tau_{n_1\!}>\tau_{n_0}$, $n_0=n_1-1$, since the cardinality of the group $G_{n_1}(q_{n_0})$ or $G_{n_1}(q_{n_1})$ is always upper-bounded by a small constant, and it is always easy to find a permutation $\pi\in G_*$ to satisfy $\pi q_{n_1} \in {\cal T}(;q_{n_0},\tau_{n_0})$ such that $\rho\kern0.1em(q_{n_1},\tau_{n_1};{\cal F}q_{n_0},\tau_{n_0}) = (-1)^{\pi} \rho\kern0.1em(\pi q_{n_1},\tau_{n_1};{\cal F}q_{n_0},\tau_{n_0}) = (-1)^{\pi} \rho_{\mathsmaller{\in}}(\pi q_{n_1},\tau_{n_1};q_{n_0},\tau_{n_0})$, or to fulfill $\pi q_{n_0} \in {\cal T}(q_{n_1},\tau_{n_1};)$ such that $\rho\kern0.1em({\cal F}q_{n_1},\tau_{n_1};q_{n_0},\tau_{n_0}) = (-1)^{\pi} \rho\kern0.1em({\cal F}q_{n_1},\tau_{n_1};\pi q_{n_0},\tau_{n_0}) = (-1)^{\pi} \rho_{\mathsmaller{\ni}}(q_{n_1},\tau_{n_1};\pi q_{n_0},\tau_{n_0})$, where $\rho_{\mathsmaller{\in}}(\pi q_{n_1},\tau_{n_1};q_{n_0},\tau_{n_0})$ or $\rho_{\mathsmaller{\ni}}(q_{n_1},\tau_{n_1};\pi q_{n_0},\tau_{n_0})$ can be efficiently simulated by Monte Carlo for the forward or backward restricted path integral, so long as the nodal surface $\partial{\cal T}(;q_{n_0},\tau_{n_0})$ or $\partial{\cal T}(q_{n_1},\tau_{n_1};)$ is known or efficiently computable. Indeed, for any CFF interaction $H_{n_1\|\mathsmaller{K}}$, $n_1\in\mathbb{N}$ and a fixed $q\in{\cal C}$, the eigen system of $H_{n_1\|\mathsmaller{K}}$ or an associated Gibbs operator $e^{-\tau H_{n_1\|\mathsmaller{K}}\!}$, $\tau>0$, restricted to the configuration space ${\cal C}_{n_1\|\mathsmaller{K}}(q)$, can be solved either analytically or numerically at a constant computational cost, which enables efficient computation of any related fermionic Gibbs kernel and nodal surfaces.

Now consider to compute the Gibbs transition amplitude $\rho\kern0.1em(q_{\mathsmaller{N}},\tau_{\mathsmaller{N}};{\cal F}q_0,\tau_0)$ for any given $(q_{\mathsmaller{N}},q_0)\in{\cal C}^2$ with respect to a full Feynman stack associated with a sequence of Gibbs operators $\{{\sf G}_n \defeq e^{-\delta\tau H_{n\|\mathsmaller{K}}}\}_{n\in[1,\,N]}$, $N\in\mathbb{N}$. By the identity $\rho\kern0.1em(\pi q_{\mathsmaller{N}},\tau_{\mathsmaller{N}};{\cal F}q_0,\tau_0) = (-1)^{\pi} \rho\kern0.1em(\pi q_{\mathsmaller{N}},\tau_{\mathsmaller{N}};{\cal F}q_0,\tau_0)$, $\forall(q_{\mathsmaller{N}},q_0)\in{\cal C}^2$, $\forall\pi\in G_*$ and the fact that the orbit of any nodal cell tiles up the configuration space ${\cal C}$ under the group action of $G_*$ \cite{Ceperley91,Wei20}, it is sufficient and WLOG to limit $q_{\mathsmaller{N}}$ to the nodal cell ${\cal N}(\tau_{\mathsmaller{N}};q_0,\tau_0)$, which is the connected component of the open set $\{q\in{\cal C} : \rho\kern0.1em(q,\tau_{\mathsmaller{N}};{\cal F}q_0,\tau_0) > 0\}$ that contains the point $q=q_0$. On the other hand, the identity $\rho\kern0.1em({\cal E}q_{\mathsmaller{N}},\tau_{\mathsmaller{N}};{\cal F}q_0,\tau_0) \defeq |A_*|^{-1}\sum_{\pi\in A_*} \rho\kern0.1em(\pi q_{\mathsmaller{N}},\tau_{\mathsmaller{N}};{\cal F}q_0,\tau_0) = \rho\kern0.1em(q_{\mathsmaller{N}},\tau_{\mathsmaller{N}};{\cal F}q_0,\tau_0)$, $\forall(q_{\mathsmaller{N}},q_0)\in{\cal C}^2$ can be used freely to multiply an point $q_{\mathsmaller{N}}$ into an orbit of end points $A_*q_{\mathsmaller{N}}$, all of which lead to the same-valued fermionic Gibbs transition amplitude. The same can be done for a start point. Such multiplication of an end or start point to its $A_*$-orbit is called {\it alternating broadcast}.

By Feynman's rule of amplitude multiplication for events occurring in succession \cite{Feynman65}, which may be viewed as a generalization of the Chapman-Kolmogorov equation in probability theory to signed densities, the fermionic Gibbs transition amplitude $\rho\kern0.1em(q_{\mathsmaller{N}},\tau_{\mathsmaller{N}};{\cal F}q_0,\tau_0)$ can be computed as
\begin{align}
& \rho\kern0.1em(q_{\mathsmaller{N}},\tau_{\mathsmaller{N}};{\cal F}q_0,\tau_0) \nonumber \\[0.75ex]
\,=\,\;& \scalebox{1.25}{$\int$}_{\!q_{1\,} \in\, {\cal C}} \; \rho\kern0.1em(q_{\mathsmaller{N}},\tau_{\mathsmaller{N}};q_1,\tau_1) \, \rho\kern0.1em(q_1,\tau_1;{\cal F}q_0,\tau_0) \, dq_1 \nonumber \\[0.75ex]
\,=\,\;& \scalebox{1.25}{$\int$}_{\!q_{1\,} \in\, G_*{\cal R}_1(q_0)} \; \rho\kern0.1em(q_{\mathsmaller{N}},\tau_{\mathsmaller{N}};q_1,\tau_1) \, \rho\kern0.1em(q_1,\tau_1;{\cal F}q_0,\tau_0) \, dq_1 \nonumber \\[0.75ex]
\,=\,\;& C_1(q_0) \, \scalebox{1.35}{$\sum$}_{\,\pi \in G_*} \, \scalebox{1.25}{$\int$}_{\!q_1 \in\, {\cal N}_1(q_0)} \; \rho\kern0.1em(q_{\mathsmaller{N}},\tau_{\mathsmaller{N}};\pi q_1,\tau_1) \, \rho\kern0.1em(\pi q_1,\tau_1;{\cal F}q_0,\tau_0) \, dq_1 \label{ChapKolmInsertF} \\[0.75ex]
\,=\,\;& C_1(q_0) \, \scalebox{1.35}{$\sum$}_{\,\pi \in G_*} \, \scalebox{1.25}{$\int$}_{\!q_1 \in\, {\cal N}_1(q_0)} \; \rho\kern0.1em(q_{\mathsmaller{N}},\tau_{\mathsmaller{N}};\pi q_1,\tau_1) \, (-1)^{\pi} \rho\kern0.1em(q_1,\tau_1;{\cal F}q_0,\tau_0) \, dq_1 \nonumber \\[0.75ex]
\,=\,\;& C_1(q_0) \, |G_*| \, \scalebox{1.25}{$\int$}_{\!q_1 \in\, {\cal N}_1(q_0)} \; \rho\kern0.1em(q_{\mathsmaller{N}},\tau_{\mathsmaller{N}};{\cal F}q_1,\tau_1) \, \rho\kern0.1em(q_1,\tau_1;{\cal F}q_0,\tau_0) \, dq_{1\,}, \nonumber
\end{align}
where ${\cal R}_1(q_0)$ is the boltzmannonic reach of $q_0$ by $H_{1\|\mathsmaller{K}}$, $G_*{\cal R}_1(q_0) \defeq \bigcup\,\{\pi\,{\cal R}_1(q_0):\pi\in G_*\}$ is the set of points that are boltzmannonically reachable from any point in the orbit $G_*q_0$, ${\cal N}_1(q_0) \defeq {\cal N}(\tau_1;q_0,\tau_0)$, the third equality follows from the fact that ${\cal N}_1(q_0)$ tiles up $G_*{\cal R}_1(q_0)$ under the group action of $G_*$ \cite{Ceperley91,Wei20}, and $C_1^{-1}(q_0) \defeq |G_*| \, V_g({\cal N}_1(q_0)) \,/\, V_g(G_*{\cal R}_1(q_0))$ is an integer counting how many times each point $q \in G_*{\cal R}_1(q_0)$ almost surely is covered by the orbit of nodal cells $\{\pi \, {\cal N}_1(q_0) : \pi \in G_*\}$, which is efficiently computable for any $q_0\in{\cal C}$ since $H_{1\|\mathsmaller{K}}$ is a CFF interaction. Then one uses alternating broadcast and proceeds as
\begin{align}
& \rho\kern0.1em(q_{\mathsmaller{N}},\tau_{\mathsmaller{N}};{\cal F}q_0,\tau_0) \,=\,
\rho\kern0.1em({\cal E}q_{\mathsmaller{N}},\tau_{\mathsmaller{N}};{\cal F}q_0,\tau_0) \nonumber \\[0.75ex]
\,=\,\;& C_1(q_0) \, |G_*| \, \scalebox{1.25}{$\int$}_{\!q_1 \in\, {\cal N}_1(q_0)} \; \rho\kern0.1em({\cal E}q_{\mathsmaller{N}},\tau_{\mathsmaller{N}};{\cal F}q_1,\tau_1) \, \rho\kern0.1em(q_1,\tau_1;{\cal F}q_0,\tau_0) \, dq_1 \nonumber \\[0.75ex]
\,=\,\;& C_1(q_0) \, |G_*| \, \scalebox{1.25}{$\int$}_{\!q_1\in\, {\cal N}_1(q_0)} \; \rho\kern0.1em({\cal E}q_{\mathsmaller{N}},\tau_{\mathsmaller{N}};{\cal F}q_1,\tau_1) \, \rho\kern0.1em({\cal E}q_1,\tau_1;{\cal F}q_0,\tau_0) \, dq_1 \\[0.75ex]
\,=\,\;& 2 \, C_1(q_0) \, \scalebox{1.35}{$\sum$}_{\,\pi_1 \in A_*} \, \scalebox{1.25}{$\int$}_{\!q_1 \in\, {\cal N}_1(q_0)} \; \rho\kern0.1em({\cal E}q_{\mathsmaller{N}},\tau_{\mathsmaller{N}};{\cal F}q_1,\tau_1) \, \rho\kern0.1em(\pi_1q_1,\tau_1;{\cal F}q_0,\tau_0) \, dq_1 \nonumber \\[0.75ex]
\,=\,\;& 2 \, C_1(q_0) \, \scalebox{1.35}{$\sum$}_{\,\pi_1 \in A_*} \, \scalebox{1.25}{$\int$}_{\!q_1 \in\, {\cal N}_1(\pi_1q_0)} \; \rho\kern0.1em({\cal E}q_{\mathsmaller{N}},\tau_{\mathsmaller{N}};{\cal F}q_1,\tau_1) \, \rho_{\mathsmaller{\in}}(q_1,\tau_1;\pi_1q_0,\tau_0) \, dq_{1\,}, \nonumber
\end{align}
so to have the first Feynman slab {\it path-rectified}, namely, represented by equivalent, non-negative definite, restricted path integrals $\rho_{\mathsmaller{\in}}(q_1,\tau_1;\pi_1q_0,\tau_0)$, each of which with a $\pi_1\in A_*$ and a $q_1 \in\, {\cal N}_1(\pi_1q_0)$ sums up non-negative Wiener densities of ${\cal T}(;\pi_1q_0,\tau_0)$-restricted paths.

It is straightforward to repeat the same procedure inductively and have all of the Feynman slabs path-rectified, so that $\rho\kern0.1em(q_{\mathsmaller{N}},\tau_{\mathsmaller{N}};{\cal F}q_0,\tau_0)$ becomes an integral of all non-negative definite contributions as
\begin{align}
& \rho\kern0.1em(q_{\mathsmaller{N}},\tau_{\mathsmaller{N}};{\cal F}q_0,\tau_0) \,=\,
\rho\kern0.1em({\cal E}q_{\mathsmaller{N}},\tau_{\mathsmaller{N}};{\cal F}q_0,\tau_0) \nonumber \\[0.75ex]
\,=\,\;& \scalebox{1.5}{$\{$} 2 \, C_n(q_{n-1}) \; \scalebox{1.35}{$\sum$}_{\,\pi_n \in A_*} \, \scalebox{1.25}{$\int$}_{\!q_n \in\, {\cal N}_n(\pi_nq_{n-1})} \scalebox{1.5}{$\}$}_{n=1}^{\mathsmaller{N}-1} \; \rho\kern0.1em({\cal E}q_{\mathsmaller{N}},\tau_{\mathsmaller{N}};{\cal F}q_{\mathsmaller{N}-1},\tau_{\mathsmaller{N}-1}) \nonumber \\[0.75ex]
\,\times\,\;& \scalebox{1.25}{$\prod$}_{n=1}^{\mathsmaller{N}-1} \, \rho_{\mathsmaller{\in}}(q_n,\tau_n;\pi_nq_{n-1},\tau_{n-1}) \; \scalebox{1.25}{$\prod$}_{n=1}^{\mathsmaller{N}-1} \, dq_n \label{FinalRPI} \\[0.75ex]
\,=\,\;& \scalebox{1.5}{$\{$} 2 \, C_n(q_{n-1}) \; \scalebox{1.35}{$\sum$}_{\,\pi_n \in A_*} \, \scalebox{1.25}{$\int$}_{\!q_n \in\, {\cal N}_n(\pi_nq_{n-1})} \scalebox{1.5}{$\}$}_{n=1}^{\mathsmaller{N}-1} \; \rho_{\mathsmaller{\in}}(q_{\mathsmaller{N}},\tau_{\mathsmaller{N}};\pi_{\mathsmaller{N}}q_{\mathsmaller{N}-1},\tau_{\mathsmaller{N}-1}) \nonumber \\[0.75ex]
\,\times\,\;& \scalebox{1.25}{$\prod$}_{n=1}^{\mathsmaller{N}-1} \, \rho_{\mathsmaller{\in}}(q_n,\tau_n;\pi_nq_{n-1},\tau_{n-1}) \; \scalebox{1.25}{$\prod$}_{n=1}^{\mathsmaller{N}-1} \, dq_{n\,}, \nonumber
\end{align}
where ${\cal N}_n(q_{n-1}) \defeq {\cal N}(\tau_n;q_{n-1},\tau_{n-1})$ denotes the nodal cell of $\rho\kern0.1em(\cdot,\tau_n;{\cal F}q_{n-1},\tau_{n-1})$ containing the point $q_{n-1}$, $C_n^{-1}(q_{n-1}) \defeq |G_*| \, V_g({\cal N}_n(q_{n-1})) \,/\, V_g(G_*{\cal R}_n(q_{n-1}))$ is an efficiently computable integer counting how many times each point $q \in G_*{\cal R}_n(q_{n-1})$ is covered by the orbit of nodal cells $\{\pi \, {\cal N}_n(q_{n-1}) : \pi \in G_*\}$, with ${\cal R}_n(q_{n-1})$ denoting the boltzmannonic reach of $q_{n-1}$ by $H_{n\|\mathsmaller{K}}$, for all $n\in[1,N]$ and any $q_{n-1}\in{\cal C}$. It is worth noting that all path rectifications in equation (\ref{FinalRPI}) are done on a per Feynman slab basis, only requiring a solution for the nodal surface of the CFF interaction associated with each Feynman slab, which can be obtained at no more than a polynomial computational cost to within a polynomial accuracy, because every CFF interaction moves no more than a small constant number of degrees of freedom around any given configuration point. An alternative derivation of equation (\ref{FinalRPI}) uses equation (\ref{ChapKolmInsertF}) repeatedly for each of the Feynman slabs indexed by $n\in[1,N{-}1]$ and takes advantage of the idempotency property of the exchange symmetrization operator to insert an ${\cal F}$ to the start point of each of the Gibbs transition amplitudes associated with the Feynman slabs \cite{Dornheim15,Yan17} and obtain an identity
\begin{align}
\rho\kern0.1em(q_{\mathsmaller{N}},\tau_{\mathsmaller{N}};{\cal F}q_0,\tau_0) \,=\,\;& \scalebox{1.5}{$\{$} C_n(q_{n-1}) \; |G_*| \, \scalebox{1.25}{$\int$}_{\!q_n \in\, {\cal N}_n(q_{n-1})} \scalebox{1.5}{$\}$}_{n=1}^{\mathsmaller{N}-1} \; \rho\kern0.1em(q_{\mathsmaller{N}},\tau_{\mathsmaller{N}};{\cal F}q_{\mathsmaller{N}-1},\tau_{\mathsmaller{N}-1}) \\[0.75ex]
\,\times\,\;& \scalebox{1.25}{$\prod$}_{n=1}^{\mathsmaller{N}-1} \, \rho\kern0.1em(q_n,\tau_n;{\cal F}q_{n-1},\tau_{n-1}) \; \scalebox{1.25}{$\prod$}_{n=1}^{\mathsmaller{N}-1} \, dq_{n\,}, \nonumber
\end{align}
then path-rectifies each of the Feynman slabs, replaces $\rho\kern0.1em(q_n,\tau_n;{\cal F}q_{n-1},\tau_{n-1})$ by $\rho_{\mathsmaller{\in}}(q_n,\tau_n;q_{n-1},\tau_{n-1})$ in accordance with equation (\ref{RhoIn}), and uses alternating broadcast, for all $n\in[1,N{-}1]$.

Equation (\ref{FinalRPI}) represents $\rho\kern0.1em(q_{\mathsmaller{N}},\tau_{\mathsmaller{N}};{\cal F}q_0,\tau_0)$ by an RPI comprising restricted Feynman paths of the form $\gamma_{\mathsmaller{N}} * \gamma_{\mathsmaller{N}-1} * \cdots * \gamma_2 * \gamma_1$ with $\gamma_n \in \Gamma_{\mathsmaller{\in}}(q_n,\tau_n;q_{n-1},\tau_{n-1})$, $q_n \in\, {\cal N}_n(\pi_nq_{n-1})$, $\pi_n\in A_*$ for all $n\in[1,N]$, where the restricted Feynman paths appear to undergo abrupt coordinate jumps in the space ${\cal C}\times[\tau_0,\tau_{\mathsmaller{N}}]$ due to the frequent insertion of even permutations. In reality, such apparent coordinate jumps do not represent actual physical discontinuities, since all points in any orbit $A_*q$, $q \in {\cal C}$ are physically equivalent and represent the same physical reality. Indeed, the restricted Feynman paths are actually continuous in the space $({\cal C}/A_*)\times[\tau_0,\tau_{\mathsmaller{N}}]$, where ${\cal C}/A_*$ is an orbifold regarding each orbit $A_*q$, $q \in {\cal C}$ as a single point. A practical and effective means to incorporate such exchange equivalence is to have each point on a Feynman plane at time $\tau_n$ associated with a couple $(\pi_n,q_n)$, $n\in[0,N]$, which specify two equivalent coordinates $q_n\in{\cal C}$ and $\pi_nq_n\in A_*q_n$, with $q_n$ serving the Feynman stack or slice from $\tau_{n-1}$ to $\tau_n$, and $\pi_nq_n$ being used by the Feynman stack or slice from $\tau_n$ to $\tau_{n+1}$. An even permutation $\pi_n\in A_*$ could be interpreted as the effect of a fermionic Gibbs kernel associated with an infinitesimally thin Feynman slice \cite{Militzer00}.

The insertion of a finite number of Feynman planes, here being done naturally for a Feynman stack associated with a sequence of Gibbs operators $\{{\sf G}_n \defeq e^{-\delta\tau H_{n\|\mathsmaller{K}}}\}_{n\in[1,\,N]}$, $N\in\mathbb{N}$, turns a Feynman path integral (or functional integration) into a finite-dimensional integral over a so-called {\it cylinder set} ${\sf Cyl} \defeq \{(q_{\mathsmaller{N}},\cdots\!,q_n,\cdots\!,q_0) : q_{n\!} \in {\cal M}, \, \forall n\in[0,N]\}$ consisting of {\it cylinder points}, where ${\cal M}={\cal C}/A_*$ or ${\cal M}={\cal C}$ depending upon whether alternating broadcast is used, every cylinder point of the form $(q_{\mathsmaller{N}},\cdots\!,q_n,\cdots\!,q_0)$ actually represents {\it a series of connected Feynman spindles}, each of which as specified by a pair of consecutive coordinates $(q_n,q_{n-1})\in{\cal M}^2$, $n\in[1,N]$ integrates under a prescribed Wiener measure into a Gibbs transition amplitude $\rho\kern0.1em(q_n,\tau_n;{\cal F}q_{n-1},\tau_{n-1})$. A {\it restricted cylinder set} ${\sf ResCyl}$ is a subset of ${\sf Cyl}$ consisting of {\it restricted cylinder points} (RCPs) of the form $(q_{\mathsmaller{N}},\cdots\!,q_n,\cdots\!,q_0)$ subject to the constraint that $\forall n\in[1,N]$, $q_{n\!} \in {\cal N}_n(q_{n-1})$, or $q_{n\!} \in A_{*\,}{\cal N}_n(q_{n-1})$ when alternating broadcast is used, in which case a (restricted) cylinder point may be specified by a particular representative of the form $(\pi_{\mathsmaller{N}},q_{\mathsmaller{N}};\cdots\!;\pi_n,q_n;\cdots\!;\pi_0=1,q_0)$, with each couple $(\pi_n,q_n)\in\in{\cal C}/A_*$, $n\in[0,N]$ representing a configuration point on the quotient manifold ${\cal C}/A_*$.

Equations (\ref{RhoIn}), (\ref{RhoNi}), and (\ref{FinalRPI}) provide a general method for simulating any quantum SCFF system or Hamiltonian on a classical computer using Monte Carlo without a numerical sign problem. A Monte Carlo procedure may employ either a homogeneous Markov chain sampling an RCP as a whole from the cylinder set $({\cal C}/A_*)^{\mathsmaller{N}+1}$, or an inhomogeneous Markov chain driving a random walk in discrete time $n\in[0,N]$ over the configuration space ${\cal C}/A_*$, involving a Markov transition associated with $\rho_{\mathsmaller{\in}}(\cdot,\tau_n;\cdot,\tau_{n-1})$ for each $n\in[1,N]$, so to evolve an initial probability distribution ${\rm Pr}\kern0.1em(\pi_0=1,q_0)$, $(\pi_0,q_0) \in\in {\cal C}/A_*$ into a sequence of probability distributions $\{{\rm Pr}\kern0.1em(\pi_n,q_n;\cdots;\pi_0=1,q_0) : n\in[0,N]\}$ for Markov sample paths, with $P_0$ being essentially the positive part of a wavefunction $\psi_0$, while the probability ${\rm Pr}\kern0.1em(\pi_{\mathsmaller{N}},q_{\mathsmaller{N}};\cdots;\pi_0,q_0)$ of any Markov sample path $(\pi_{\mathsmaller{N}},q_{\mathsmaller{N}};\cdots;\pi_0,q_0)$ at the end is proportional to the Wiener density of $(\pi_{\mathsmaller{N}},q_{\mathsmaller{N}};\cdots;\pi_0,q_0)$ as an RCP. Also, a suitable boundary condition should be chosen for the Feynman planes at the two ends. One usual choice enforces periodicity by identifying the $0$-th and $N$-th Feynman planes as one and the same. Another frequent choice provides a known probability distribution in $q_0$ or $q_{\mathsmaller{N}}\in{\cal C}/A_*$ at each end.

In one exemplary embodiment of Markov chain Monte Carlo (MCMC) using a homogenous chain to sample RCPs subject to the periodic boundary condition, a suitable random coordinate $q_0 \in {\cal C}/A_*$ is chosen such that it extends into an initial RCP $\{q_n=q_0\}_{n=1}^{\mathsmaller{N}}$ as a suitable start point, then the RPI of equation (\ref{FinalRPI}) is approximated by iterating three steps of random moves to wiggle the RCP for a predetermined and ${\rm poly}({\rm size}({\cal C}))$-bounded number of times. At the start of each iteration, let $(\pi_{\mathsmaller{N}},q_{\mathsmaller{N}};\cdots\!;\pi_n,q_n;\cdots\!;\pi_0=1,q_0)$ denote the instantaneous RCP. The first step draws a random integer $n\in[1,N]$ uniformly to name the $n$-th component of the instantaneous RCP, that is a couple $(\pi_n,q_n) \in\in {\cal C}/A_*$ representing a configuration point on the $n$-th Feynman plane. Let the couples $(\pi_{(n-1)\|\mathsmaller{N}},q_{(n-1)\|\mathsmaller{N}}), \, (\pi_{(n+1)\|\mathsmaller{N}},q_{(n+1)\|\mathsmaller{N}}) \in\in {\cal C}/A_*$ be associated with the Feynman planes immediately before or after the $n$-th Feynman plane. The second step simply chooses a $\pi\in A_*$ uniformly and makes the substitutions $\pi_n \leftarrow \pi\pi_n$, $q_n \leftarrow \pi q_n$, $\pi_{(n+1)\|\mathsmaller{N}} \leftarrow \pi_{(n+1)\|\mathsmaller{N}}\pi^{-1}$. The third step  performs a random walk of $q_{n\!} \in {\cal C}$ using either Metropolis-Hastings or Gibbs sampling \cite{Metropolis53,Hastings70,Geman84} in accordance with a conditional probability
\begin{align}
{\rm Pr}\kern0.1em(q_n|\cdots) \,\defeq\,& {\rm Pr}\kern0.1em(q_n|\,\pi_{(n+1)\|\mathsmaller{N}},q_{(n+1)\|\mathsmaller{N}};\pi_n;q_{(n-1)\|\mathsmaller{N}}) \nonumber \\[0.75ex]
\,\propto\,\;& C_{(n+1)\|\mathsmaller{N}}(q_n) \, \rho_{\mathsmaller{\in}}(q_{(n+1)\|\mathsmaller{N}},\tau_{(n+1)\|\mathsmaller{N}};\pi_{(n+1)\|\mathsmaller{N}}q_n,\tau_n) \\[0.75ex]
\,\times\,\;& C_n(q_{(n-1)\|\mathsmaller{N}}) \, \rho_{\mathsmaller{\in}}(q_n,\tau_n;\pi_nq_{(n-1)\|\mathsmaller{N}},\tau_{(n-1)\|\mathsmaller{N}}) \,, \nonumber
\end{align}
which vanishes whenever $q_n \not\in {\cal N}_n(\pi_nq_{(n-1)\|\mathsmaller{N}})$ or $q_{n+1} \not\in {\cal N}_{(n+1)\|\mathsmaller{N}}(\pi_{(n+1)\|\mathsmaller{N}}q_n)$. In an exemplary sampler, a random coordinate $r_n$ is drawn according to the probability distribution ${\rm Pr}\kern0.1em(r_n|\cdots)$, $r_n \in {\cal C}$, then a substitution $q_n \leftarrow r_n$ is executed to update the coordinate. Note that the conditional probability ${\rm Pr}\kern0.1em(q_n|\cdots)$ is always efficiently computable, since all of the $C_{\cdot}(\cdot)$ and $\rho_{\mathsmaller{\in}}(\cdot)$ quantities involve CFF interactions and can be computed either analytically or numerically at a constant cost.

Another exemplary embodiment uses the method of reptation quantum Monte Carlo (RQMC) \cite{Baroni98arXiv,Baroni99prl,Carleo10,Wei20} and path integral to compute a mirror-symmetric sequence of Gibbs operators $\{{\sf G}_n\}_{n=1}^{\mathsmaller{2N}}$ with ${\sf G}_{n\!} \defeq e^{-\delta\tau H_n}$, $\forall n\in[1,N]$ and ${\sf G}_{n\!} \defeq {\sf G}_{\mathsmaller{2N}-n+1}$, $\forall n\in[N+1,2N]$, where $N=(m+1)K$, $K\in\mathbb{N}$, $m\in\mathbb{N}$, and the constant $\delta\tau \in (0,\infty)$ is sufficiently small such that $\forall l\in[0,m]$, the product operator $\prod_{\mathsmaller{n=lK+1}}^{\mathsmaller{(l+1)K}}{\sf G}_n$ is substantially the same as the Gibbs operator $\exp[-\delta\tau H(l)]$, with $H(l) \defeq \sum_{\mathsmaller{n=lK+1}}^{\mathsmaller{(l+1)K}}H_n$ being an LTK-decomposed SCFF Hamiltonian. With the boundary condition that assigns a probability density ${\rm Pr}\kern0.1em(q_0) = \max(0, \phi_0(q_0)) / D_0$ for all $q_0\in{\cal C}/A_*$ at the start on the first Feynman plane and similarly at the end, $D_0 \defeq \frac{1}{2} \int_{q\in{\cal C}/A*} |\phi_0(q)| \, dq$, $\phi_0\in L^2({\cal C}/A_*)$ being any given MFS wavefunction, the method of RQMC computes the expectation value of any $G_*$-invariant, $({\cal C}/G_*)$-diagonal operator $V$ under ${\sf G}_*|\phi_0\ra$, ${\sf G}_* \defeq \prod_{n=1}^{\mathsmaller{N}} {\sf G}_n$ through
\begin{equation}
\frac{\la\phi_0 | {\sf G}_*^+ V {\sf G}_* | \phi_0\ra} {\la\phi_0 | {\sf G}_*^+ {\sf G}_* | \phi_0\ra} \,=\, \frac{\int_{\xi\in{\sf ResCyl}} \phi_0(\xi[2N][1]) \, V(\xi[N][1]) \, W(\xi) \, d\xi} {\int_{\xi\in{\sf ResCyl}} \phi_0(\xi[2N][1]) \, W(\xi) \, d\xi} \,, \label{GVGExpValLTK}
\end{equation}
where $\xi \defeq (\pi_{\mathsmaller{2N}},q_{\mathsmaller{2N}};\cdots\!;\pi_n,q_n;\cdots\!;\pi_0=1,q_0)$ traverses the restricted cylinder set ${\sf ResCyl}$, and $\forall\xi\in{\sf ResCyl}$, $\forall n\in[0,2N]$, $\xi[n] \defeq (\pi_n,q_n)$ denotes the $n$-th component of $\xi$, $\xi[n][0] \defeq \pi_n$, $\xi[n][1] \defeq q_n$, $W(\xi)$ is the Wiener density at the RCP $\xi = (\pi_{\mathsmaller{2N}},q_{\mathsmaller{2N}};\cdots\!;\pi_n,q_n;\cdots\!;\pi_0,q_0)$ evaluated as
\begin{equation}
W(\xi) \,\defeq\, W(\pi_n,q_n;\cdots\!;\pi_0,q_0) \,\defeq\, \phi_0(q_0) \, \scalebox{1.25}{$\prod$}_{n=1}^{\mathsmaller{2N}} \scalebox{1.25}{$\{$} C_n(q_{n-1}) \, \rho_{\mathsmaller{\in}}(q_n,\tau_n;\pi_nq_{n-1},\tau_{n-1}) \scalebox{1.25}{$\}$} \,,
\end{equation}
where $\tau_n \defeq n\delta\tau$, $\forall n\in[0,2N]$, the $\int_{\xi\in{\sf ResCyl}} d\xi$ operation on the right-hand side of equation (\ref{GVGExpValLTK}) involves integrating $q_0$ over a nodal cell of $\phi_0$, then for each $n\in[1,2N]$, integrating $q_n$ over the nodal cell ${\cal N}_n(\pi_nq_{n-1})$ and summing $\pi_n$ over $A_*$. In practice, the integration $\int_{\xi\in{\sf ResCyl}} d\xi$ is of course approximated by summing up a finite number of RCPs obtained by importance sampling through an MCMC procedure.

Interestingly, the desired RCP samples can be generated by running an inhomogeneous Markov chain over the state space ${\cal C}/A_*$ when the CFF interactions satisfy a suitable condition. The inhomogeneous Markov chain starts with a random walker having an initial probability density ${\rm Pr}\kern0.1em(\pi_0=1,q_0)$, $(\pi_0,q_0) \in\in {\cal C}/A_*$, proceeds inductively in steps indexed by $n\in[1,2N]$ and records a sequence of configuration coordinates as a Markov sample path or trajectory of the random walker. At the beginning of each step $n\in[1,2N]$, the walker has reached a point $(\pi_{n-1},q_{n-1}) \in\in {\cal C}/A_*$ via a Markov sample path $(\pi_{n-1},q_{n-1};\cdots\!;\pi_0=1,q_0)$, which is associated with a probability density ${\rm Pr}\kern0.1em(\pi_{n-1},q_{n-1};\cdots\!;\pi_0,q_0)$. The walker then undergoes a Markov transition from the present position $(\pi_{n-1},q_{n-1})$ to a new coordinate $(\pi_n,q_n) \in\in {\cal C}/A_*$, $q_n \in {\cal N}_n(\pi_nq_{n-1})$ in accordance with a Markov transition probability
\begin{equation}
{\rm Pr}\kern0.1em(q_{n\,}|\,\pi_nq_{n-1}) \;=\; \frac{C_n(q_{n-1}) \, \rho_{\mathsmaller{\in}}(q_n,\tau_n;\pi_nq_{n-1},\tau_{n-1})} {D_n(q_{n-1}) \,\defeq\, \scalebox{1.2}{$\{$} C_n(q_{n-1}) \int \rho_{\mathsmaller{\in}}(r_n,\tau_n;\pi_nq_{n-1},\tau_{n-1}) \, dr_n \scalebox{1.2}{$\}$}} \,, \label{MarkovTransProb}
\end{equation}
so to extend the Markov sample path into $(\pi_n,q_n;\cdots\!;\pi_0,q_0)$ with a probability
$${\rm Pr}\kern0.1em(\pi_n,q_n;\cdots\!;\pi_0,q_0) \,=\, {\rm Pr}\kern0.1em(q_{n\,}|\,\pi_nq_{n-1}) \, {\rm Pr}\kern0.1em(\pi_{n-1},q_{n-1};\cdots\!;\pi_0,q_0) \,,$$
where $D_n(q_{n-1})$ is called the {\it amplitude integral} of $\rho\kern0.1em(\cdot,\tau_n;\pi_nq_{n-1},\tau_{n-1}) = \la\cdot|e^{-(\tau_n-\tau_{n-1})H_n}|q_{n-1}\ra$ over the nodal cell containing $\pi_nq_{n-1}\in{\cal C}$. At the end, the inhomogeneous Markov chain generates a Markov sample path $(\pi_{\mathsmaller{2N}},q_{\mathsmaller{2N}};\cdots\!;\pi_0,q_0)$ with an associated probability density
\begin{equation}
{\rm Pr}\kern0.1em(\pi_{\mathsmaller{2N}},q_{\mathsmaller{2N}};\cdots\!;\pi_0,q_0) \,=\, {\rm Pr}\kern0.1em(\pi_0,q_0) \; {\textstyle{\prod_{n=1\,}^{\mathsmaller{2N}}}} {\rm Pr}\kern0.1em(q_{n\,}|\,\pi_nq_{n-1}) \,=\, \frac{W(\pi_{\mathsmaller{2N}},q_{\mathsmaller{2N}};\cdots\!;\pi_0,q_0)} {\prod_{n=1}^{\mathsmaller{2N}} D_n(q_{n-1})} \,, \label{PvsW}
\end{equation}
which is substantially the same as the Wiener density $W(\pi_{\mathsmaller{2N}},q_{\mathsmaller{2N}};\cdots\!;\pi_0,q_0)$ at $(\pi_{\mathsmaller{2N}},q_{\mathsmaller{2N}};\cdots\!;\pi_0,q_0)$ as an RCP representing {\it a series of restricted Feynman spindles}, when the sequence of Gibbs operators $\{{\sf G}_n\}_{n=1}^{\mathsmaller{N}}$ is amplitude integral-balanced as defined below, such that, running said inhomogeneous Markov chain repeatedly and independently generates a polynomial number of Markov sample paths that are effectively RCPs and can be used to estimate to within a polynomial accuracy the expectation value of any $G_*$-invariant, $({\cal C}/G_*)$-diagonal observable $V$ according to equation (\ref{GVGExpValLTK}).

\begin{definition}{} \label{defiAIB}
A sequence of Gibbs operators $\{{\sf G}_n \defeq e^{-\tau H_n}\}_{n=1}^{\mathsmaller{N}}$, $N\in\mathbb{N}$, $\tau\in(0,\infty)$ in association with a sequence of CFF interactions $\{H_n\}_{n=1}^{\mathsmaller{N}}$ is said to be amplitude integral-balanced (AIB), when the product $\prod_{n=1}^{\mathsmaller{N}} D_n(q_{n-1})$, with $D_n(q_{n-1})$ being defined as in equation (\ref{MarkovTransProb}), always evaluates into the same constant for any RCP $(\pi_{\mathsmaller{N}},q_{\mathsmaller{N}};\cdots\!;\pi_0=1,q_0)$ with $(\pi_0,q_0) \in\in {\cal C}/A_*$, $\pi_n \in A_*$, $q_n \in {\cal N}_n(\pi_nq_{n-1})$, $\forall n\in[1,N]$.
\iftoggle{UseREVTEX}{}{\vspace{-1.5ex}}
\end{definition}

Yet another exemplary embodiment uses RQMC and path integral to compute a mirror-symmetric sequence of Gibbs operators $\{{\sf G}_n\}_{n=1}^{\mathsmaller{2N}}$ with ${\sf G}_{n\!} \defeq e^{-\tau H_n}$, $\forall n\in[1,N]$ and ${\sf G}_{n\!} \defeq {\sf G}_{\mathsmaller{2N}-n+1}$, $\forall n\in[N+1,2N]$, where $N=(m+1)m_0K$, $K\in\mathbb{N}$, $m_0\in\mathbb{N}$, $m\in\mathbb{N}$, $H_{lm_0\mathsmaller{K}+n} = H_{lm_0\mathsmaller{K}+(n\|\mathsmaller{K})}$ for all $l\in[0,m]$ and for all $n\in[1,m_0K]$, the constant $\tau = O({\rm poly}({\rm size}({\cal C}))) \in (0,\infty)$ is no longer small but sufficiently large such that $\forall n\in[1,N]$, ${\sf G}_n$ is essentially the same as $\mathit{\Pi}_n = \lim_{\,t\rightarrow\infty}e^{-t[H_n-\lambda_0(H_n)]}$ up to an error that is exponentially small, while $m_0 = O({\rm poly}({\rm size}({\cal C})))$ is sufficiently large such that $\forall l\in[0,m]$, ${\sf G}(l) \defeq \prod_{n=lm_0\mathsmaller{K}+1}^{(l+1)m_0\mathsmaller{K}}{\sf G}_n$ is essentially the same as $\mathit{\Pi}(l) = \lim_{\,t\rightarrow\infty}e^{-t\{H(l)-\lambda_0[H(l)]\}}$ up to a constant $A_{m_0}>0$ and an error that is $O(1/{\rm poly}(m_0))$, with $H(l) \defeq \sum_{\,k=lm_0\mathsmaller{K}+1}^{\,(l+1)m_0\mathsmaller{K}} H_k$ being a GSP-decomposed SCFF Hamiltonian. If the sequence of Gibbs operators $\{{\sf G}_n\}_{n=1}^{\mathsmaller{N}}$ is AIB, then it is efficiently simulatable via Monte Carlo in exactly the same manner as described above when dealing with LTK-decomposed Hamiltonians.

\begin{theorem}{} \label{TheoremOne}
Let $\{{\sf G}_n \defeq e^{-\tau H_n}\}_{n=1}^{\mathsmaller{N}}$, $N\in\mathbb{N}$, $\tau\in(0,\infty)$ be an AIB sequence of Gibbs operators associated with a sequence of SCFF Hamiltonians $\{H_n\}_{n=1}^{\mathsmaller{N}}$ supported by a compact configuration space ${\cal C}$ and satisfying $\lambda_1(H_n)-\lambda_0(H_n) = \Omega(1/{\rm poly}({\rm size}({\cal C})))$, $\forall n\in[1,N]$. There is a fully polynomial randomized approximation scheme (FPRAS) \cite{Motwani95} to estimate $\la V\ra \defeq \la\phi_0|{\sf G}_*^+V{\sf G}_*|\phi_0\ra / \la\phi_0|{\sf G}_*^+{\sf G}_*|\phi_0\ra > 0$ with ${\sf G}_* \defeq \prod_{n=1}^{\mathsmaller{N}}{\sf G}_n$, for any given MFS wavefunction $\phi_0\in L^2({\cal C}/A_*)$ and any $G_*$-invariant, $({\cal C}/G_*)$-diagonal operator $V\ge 0$.
\iftoggle{UseREVTEX}{}{\vspace{-1.5ex}}
\end{theorem}

\begin{proof}[Proof]
As specified above, running an inhomogeneous Markov chain for an $O({\rm poly}({\rm size}({\cal C}),\epsilon^{-1}))$ number of times can generate a sufficient number of RCPs to produce an estimate $V_*\in\mathbb{R}$ according to equation (\ref{GVGExpValLTK}), such that ${\rm Pr} \{ |(V_*-\la V\ra) / \la V\ra| < \epsilon \} > 2/3$. For any $n\in[1,N]$, any $(\pi_n,q_n,q_{n-1}) \in A_*\times{\cal C}^2$, since $H_n$ is an CFF interaction with $\lambda_1(H_n)-\lambda_0(H_n) = \Omega(1/{\rm poly}({\rm size}({\cal C})))$, the Markov transition probability ${\rm Pr}\kern0.1em(q_{n\,}|\,\pi_nq_{n-1})$ of equation (\ref{MarkovTransProb}) is always efficiently computable with an $O({\rm poly}(\epsilon))$ accuracy at an $O({\rm poly}(\epsilon^{-1}))$ cost either analytically or using a deterministic or randomized numerical routine, the mixing time of each said Markov transition moving no more than a small constant number of degrees of freedom is always $O({\rm poly}({\rm size}({\cal C}),\,\epsilon^{-1}))$-bounded. The overall runtime is clearly $O({\rm poly}({\rm size}({\cal C}),N,\epsilon^{-1}))$.
\iftoggle{UseREVTEX}{}{\vspace{-1.5ex}}
\end{proof}

A particular application of RQMC and \myTheorem~\ref{TheoremOne} is to simulate the ground state of a given Hamiltonian, where an AIB sequence of Gibbs operators $\{{\sf G}_n \defeq e^{-\tau H_n}\}_{n=1}^{\mathsmaller{N}}$, $N\in\mathbb{N}$, $\tau\in(0,\infty)$ is associated with sequence of SCFF Hamiltonians $\{H(l)\}_{l=0}^m$, $m\in\mathbb{N}$, which evolves adiabatically from an initial Hamiltonian $H(0)$ with a known non-degenerate ground state $\phi_0\defeq\psi_0(H(0))$ to a final Hamiltonian $H(m)$ \cite{Wei20} whose ground state $\phi_m\defeq\psi_0(H(m))$ is of interest, where each $H(l)$, $l\in[0,m]$ is an SCFF Hamiltonian with a non-degenerate and polynomial-gapped ground state, whose defining form sum $H(l) = \sum_{k=1}^{\mathsmaller{K}}H_{\mathsmaller{lm_0K+k}}$ is LTK-decomposed with $m_0=1$ or GSP-decomposed with $m_0=O({\rm poly}({\rm size}({\cal C})))\in\mathbb{N}$, while $H_{lm_0\mathsmaller{K}+n} = H_{lm_0\mathsmaller{K}+(n\|\mathsmaller{K})}$ for all $l\in[0,m]$ and all $n\in[1,m_0K]$ in both cases. Clearly, $N\defeq(m+1)m_0K$. For an LTK-decomposition, the time constant $\tau$ is sufficiently small  such that the difference between $e^{-\tau H(l)}$ and $\prod_{k=1}^{\mathsmaller{K}}e^{-\tau H_{\mathsmaller{lK+k}}}$ is sufficiently small for all $l\in[0,m]$, whereas for a GSP-decomposition, the time constant $\tau$ is sufficiently large such that the Gibbs operator ${\sf G}_n$ is exponentially close to $\mathit{\Pi}_n = \lim_{\,t\rightarrow\infty}e^{-t[H_n-\lambda_0(H_n)]}$ for all $n\in[1,N]$, while $m_0$ is sufficiently large such that ${\sf G}(l) \defeq \prod_{n=lm_0\mathsmaller{K}+1}^{(l+1)m_0\mathsmaller{K}}{\sf G}_n$ is essentially the same as $\mathit{\Pi}(l) = \lim_{\,t\rightarrow\infty}e^{-t\{H(l)-\lambda_0[H(l)]\}}$ up to a constant for all $l\in[0,m]$. Finally, $m$ is chosen sufficiently large such that $\|H(l+1)-H(l)\| = O(1/{\rm poly}(m))$ is sufficiently small comparing to $\lambda_1[H(l)]-\lambda_0[H(l)]$ for all $l\in[0,m]$. As such, the final $H(m)$ is called an {\it adiabatic-reachable SCFF Hamiltonian}. \myTheorem~\ref{TheoremOne} guarantees an FPRAS for simulating the ground state $\phi_m$ of an adiabatic-reachable SCFF Hamiltonian $H(m)$, producing a good estimate for the expectation value $\la\phi_m|V|\phi_m\ra / \la\phi_m|\phi_m\ra$ of any $G_*$-invariant, $({\cal C}/G_*)$-diagonal observable $V\ge 0$.

Not only SCFF Hamiltonians can be simulated efficiently, but also they are universal for many-body physics and quantum computing. For any $\theta\in[-\pi,\pi)$, let ${\sf R}(\theta) \defeq I\cos\theta + XZ\sin\theta$ denote a {\it rotation gate} and $R(\theta) \defeq Z\cos\theta + X\sin\theta$ denote an {\it $R$-gate} \cite{Biamonte08,Wei20}, where $X \defeq \sigma_x$ and $Z \defeq \sigma_z$ are the familiar Pauli matrices acting on a single rebit as the simplest quantum system $({\cal C}_0,{\cal H}_0,{\cal B}_0)$, with ${\cal C}_0 \defeq \{0,1\}$, ${\cal H}_0 \defeq \{\alpha|0\ra + \beta|1\ra : \alpha,\beta\in\mathbb{R}\}$, ${\cal B}_0$ being the Banach algebra of $2\times 2$ real matrices. Define $Z^{\pm} \defeq (I \pm Z)/2$. It is well known that {\it controlled rotation gates} of the form $I\otimes Z^{+\!} + {\sf R}(\theta_0)\otimes Z^-$ all using the same angle $\theta_0\in[-\pi,\pi)$ are already universal for quantum computation, when $\theta_0/\pi$ is irrational \cite{Rudolph02,Wei20}. In this presentation, any universal quantum algorithm is allowed to use any controlled rotation gate ${\sf U}(\theta) \defeq I\otimes Z^{+\!} + {\sf R}(\theta)\otimes Z^-$ with any angle $\theta\in[-\pi,\pi)$, each of which is realized by a {\it controlled $R$-gate} $R(\theta/2)\otimes Z^+ + R(-\theta/2)\otimes Z^-$ followed by a {\it free $R$-gate} $R(\theta/2) = R(\theta/2)\otimes I$. The following will show that any BQP algorithm given as an ordered sequence of free or controlled $R$-gates $\{U_t(\theta_t)\defeq R_{i(t)}(\theta_t)\mbox{\;or\;}R_{i(t)}(\theta_t)\otimes Z^+_{j(t)} + R_{i(t)}(-\theta_t)\otimes Z^-_{j(t)}\}_{t\in[1,\mathsmaller{T}]}$, $T\in\mathbb{N}$ on a quantum computer of $n\in\mathbb{N}$ rebits can be mapped to an LTK- or GSP-decomposed SCFF Hamiltonian generating an AIB sequence of Gibbs operators, where each $R$-gate $R_{i(t)}(\theta_t)$, $\theta_t\in[-\pi,\pi)$ acts on a rebit indexed by an $i(t)\in[1,n]$, and $Z^{\pm}_{j(t)}$ operate on a control rebit indexed by a $j(t)\in[1,n]$, $\forall t\in[1,T]$. Such a BQP algorithm, or its associated quantum circuit, is said to have a computational size $T+n$.

\begin{definition}{} \label{defiHomoIsophysics}
A homophysics $\mathfrak{M}:({\cal C}, {\cal H}, {\cal B}) \mapsto ({\cal C}', {\cal H}', {\cal B}')$ between two quantum systems with Hamiltonians $H \in {\cal B}$ and $H' \in {\cal B}'$ is an injective mapping that sends any subset ${\cal D} \subseteq {\cal C}$ to a unique ${\cal D}' \defeq \mathfrak{M}({\cal D}) \subseteq {\cal C}'$, maps any $\psi \in {\cal H}$ to a unique $\psi' \defeq \mathfrak{M}(\psi) \in {\cal H}'$, and sends any $Q \in {\cal B}$ to a unique $Q' \defeq \mathfrak{M}(Q) \in {\cal B}'$, such that ${\cal C} \supseteq {\cal D} \mapsto \mathfrak{M}({\cal D}) \subseteq {\cal C}'$ embeds the Boolean algebra of subsets \cite{Halmos74} of ${\cal C}$ into the Boolean algebra of subsets of ${\cal C}'$; 2) ${\cal H} \ni \psi \mapsto \mathfrak{M}(\psi) \in {\cal H}'$ embeds the Hilbert space ${\cal H}$ into ${\cal H}'$; 3) ${\cal B} \ni Q \mapsto \mathfrak{M}(Q) \in {\cal B}'$ embeds the Banach algebra ${\cal B}$ into ${\cal B}'$; 4) there exists a constant $c>0$, $c+c^{-1} = O({\rm poly}({\rm size}(H)))$, with which $\la\mathfrak{M}(\psi)|\mathfrak{M}(Q)|\mathfrak{M}(\phi)\ra = c\la\psi|Q|\phi\ra$ holds $\forall \psi, \phi \in {\cal H}$, $\forall Q \in {\cal B}$; 5) ${\rm size}(H) = O({\rm poly}({\rm size}(H')))$ and ${\rm size}(H') = O({\rm poly}({\rm size}(H)))$. A homophysics $\mathfrak{M}$ is called an isophysics when the mapping $\mathfrak{M}$ is also surjective.
\iftoggle{UseREVTEX}{}{\vspace{-1.5ex}}
\end{definition}

Firstly, it is useful to construct a {\it bi-fermion} system $({\cal C}_1,{\cal H}_1,{\cal B}_1)$ consisting of two non-interacting identical fermions moving on a circle $\mathbb{T} \defeq \mathbb{R}/2\mathbb{Z}$ \cite{Wei20}, governed by a single-particle Hamiltonian $-(1/2)\partial^{2\!}/\partial x^2 \,{+}\, V(x)$, $x \in \mathbb{T}$, with an external potential $V(x) = V_0 \Iver{d(x,0)\,{>}\,1-a_0} \,{-}\, V_0 \Iver{d(x,0)\,{<}\,a_0}$, $x\in[-1,1) \! \pmod 2 \simeq\mathbb{T}$, $a_0 = \gamma_0^{-1}$, $V_0 = \gamma_0^2$, $\gamma_0\gg 1$ being a large constant, where $d(x,y)$ denotes the geodesic distance between $x\in\mathbb{T}$ and $y\in\mathbb{T}$ along the circle, $\Iver{\cdot}$ is an Iverson bracket \cite{IversonBracketWiki} which returns a number valued to $1$ or $0$ depending on if the Boolean expression inside the bracket is true or false. When $\gamma_0$ is sufficiently large, the potential well and barrier become essentially Dirac deltas, $V(x) \simeq \gamma_0\delta(x+1) \,{-}\, \gamma_0\delta(x)$, $x\in[-1,1) \! \pmod 2$, such that a bi-fermion under a nominal Hamiltonian ${\sf H}_{\mathsmaller{\rm BF}} = (\gamma_0^2\,{-}\,\pi^2)/2 + \sum_{i=1}^2[-(1/2)\partial^{2\!}/\partial x_i^2 \,{+}\, V(x_i)]$, $(x_1,x_2) \in \mathbb{T}^2$ behaves like a rebit with two low-energy states
\begin{align}
\psi_+(x_1,x_2) \,=\,\;& (1\,{-}\,\pi_{12}) \sin\pi[d(x_1,0)\,{-}\,a_0] \, e^{-\gamma_0d(x_2,0)} , \; (x_1,x_2) \in [-1,1)^2 , \\[0.75ex]
\psi_-(x_1,x_2) \,=\,\;& (1\,{-}\,\pi_{12}) \sin\pi x_1 \, e^{-\gamma_0d(x_2,0)} , \; (x_1,x_2) \in [-1,1)^2 ,
\end{align}
that are degenerate at $E_0 = 0$, where $\pi_{12}$ is the fermion exchange operator swapping the particle labels $1$ and $2$. Choose $\alpha_0 = 2\gamma_0^{-1}\log\gamma_0$, then for all $x$ such that $d(x,0) \,>\, \alpha_0$, the amplitude of the single-particle bound state $|e^{-\gamma_0d(x,0)}| \,<\, \gamma_0^{-2}$, which is rather small. Construct potential functions
\begin{align}
{\sf X}(x_1,x_2) \,=\,\;& \gamma_0 \, (1\,{+}\,\pi_{12}) \Iver{\,d(x_1,0) \,{>}\, 1-a_0 \,\wedge\, d(x_2,0) \,{<}\, \alpha_0\,} - (\pi^2/4\gamma_0^2) \, , \\[0.75ex]
{\sf Z}^+(x_1,x_2) \,=\,\;& (1\,{+}\,\pi_{12}) \Iver{\,d(x_1,+1/2) \,{<}\, 1/2 \,\wedge\, d(x_1,0) \,{>}\, \alpha_0 \,\wedge\, d(x_2,0) \,{<}\, \alpha_0\,} , \\[0.75ex]
{\sf Z}^-(x_1,x_2) \,=\,\;& (1\,{+}\,\pi_{12}) \Iver{\,d(x_1,-1/2)<1/2 \,\wedge\, d(x_1,0) \,{>}\, \alpha_0 \,\wedge\, d(x_2,0) \,{<}\, \alpha_0\,} ,
\end{align}
$\forall (x_1,x_2) \in [-1,1)^2 \simeq \mathbb{T}^2$, which can be regarded as $\{1,\pi_{12}\}$-invariant, $\mathbb{T}^2$-diagonal operators. It is straightforward to verify that such a bi-fermion implements a rebit \cite{Wei20} through a homophysics $\mathfrak{M}_1 : ({\cal C}_0, {\cal H}_0, {\cal B}_0) \mapsto ({\cal C}_1, {\cal H}_1, {\cal B}_1)$ such that, with $|\pm\ra \defeq (|0\ra \pm |1\ra) / \sqrt{2}$,
\begin{align}
\mathfrak{M}_{1\,}(|\pm\ra \in {\cal H}_0) \,=\,\;& \psi_{\pm}(x_1,x_2) \in {\cal H}_{1\,} , \\[0.75ex]
\mathfrak{M}_{1\,}(X \in {\cal B}_0) \,=\,\;& (2\gamma_0^2/\pi^2) \, [{\sf H}_{\mathsmaller{\rm BF}} + {\sf X}(x_1,x_2)] \in {\cal B}_{1\,} , \\[0.75ex]
\mathfrak{M}_{1\,}(Z^{\pm\!} \in {\cal B}_0) \,=\,\;& \gamma_0{\sf H}_{\mathsmaller{\rm BF}} + {\sf Z}^{\pm}(x_1,x_2) \in {\cal B}_{1\,} .
\end{align}
Via linear combinations, the operators $\mathfrak{M}_1(X)$, $\mathfrak{M}_1(Z^+)$, $\mathfrak{M}_1(Z^-)$ generate all partial Hamiltonians that are of interest for quantum computing on a single bi-fermion, because ${\rm span}\{X,Z^+,Z^-\}$ contains all Hermitian elements in ${\cal B}_0$. It is noted in passing that, although it is preferred for the single-particle potential $V(x)$, $x \in \mathbb{T}$ to have a narrow and deep potential well around $x = 0$, approximating a fairly strong Dirac delta to localize one of the two fermions in a small neighborhood of $x = 0$, there is no practical necessity other than convenience of mathematical analysis, to require a steep potential barrier around $x=\pm 1$. Rather, it is perfectly fine to place a relatively wide and low potential barrier, as long as its width and height are chosen properly to be commensurate with the Delta-like potential well around $x=0$, such that the nominal bi-fermion Hamiltonian $H_{\mathsmaller{\rm BF}}$ defines a degenerate two-state Hilbert space implementing a rebit.

Next, it is straightforward to construct a homophysics $\mathfrak{M}_2 : ({\cal C}_0^2, {\cal H}_0^2, {\cal B}_0^2) \mapsto ({\cal C}_1^2, {\cal H}_1^2, {\cal B}_1^2)$, with ${\cal C}_i^2 \defeq {\cal C}_i \times {\cal C}_i$, ${\cal H}_i^2 \defeq {\cal H}_i \otimes {\cal H}_i$, ${\cal B}_i^2 \defeq {\cal B}_i \otimes {\cal B}_i$, $\forall i \in \{0,1\}$, so to implement a pair of interacting rebits using two bi-fermions conditioned and interacting through the following partial Hamiltonians,
\begin{align}
\mathfrak{M}_{2\,}(X_1\otimes Z_2^{\pm} ) \,=\,\;& (2\gamma_0^2/\pi^2) \, [\, {\sf H}_{\mathsmaller{\rm BF},1} + {\sf H}_{\mathsmaller{\rm BF},2} + {\sf X}(x_{11},x_{12}) \, {\sf Z}^{\pm}(x_{21},x_{22}) \,] \,, \label{CtrlX} \\[0.75ex]
\mathfrak{M}_{2\,}(Z_1^{\pm}\otimes Z_2^{\pm} ) \,=\,\;& \gamma_0 {\sf H}_{\mathsmaller{\rm BF},1} + \gamma_0 {\sf H}_{\mathsmaller{\rm BF},2} + {\sf Z}^{\pm}(x_{11},x_{12}) \, {\sf Z}^{\pm}(x_{21},x_{22}) \,, \label{CtrlZ}
\end{align}
where $\forall i \in \{1,2\}$, $X_i$, $Z_i^{\pm}$, $Z_i = Z_i^+ - Z_i^-$ are the $X$- and $Z$-gates on the $i$-th rebit, ${\sf H}_{\mathsmaller{\rm BF},i}$ is the nominal Hamiltonian of the $i$-th bi-fermion, $(x_{i1},x_{i2}) \in \mathbb{T}^2$ is the two-fermion configuration of the $i$-th bi-fermion. $X_1\otimes Z_2^{\pm}$ and $Z_1\otimes Z_2^{\pm}$ are called {\it single-rebit-controlled gates}, whose linear combinations include all single-rebit-controlled $R$ gates, which are already universal for ground state quantum computation (GSQC) \cite{Feynman85,Kitaev02,Mizel04,Kempe06,Wei20} in the sense that, using the so-called perturbative gadgets, up to an error tolerance $\epsilon > 0$, the low-energy physics of any system of $n\in\mathbb{N}$ rebits under a computationally $k$-local Hamiltonian, $k\in\mathbb{N}$ being a fixed number, can be homophysically mapped to the low-energy physics of another system of ${\rm poly}(n,\epsilon^{-1})$ rebits under a Hamiltonian that involves only one-body and two-body interactions, especially the controlled $R$-gates, whose operator norms are upper-bounded by ${\rm poly}(\epsilon^{-1})$ \cite{Biamonte08,Jordan08,Bravyi08prl,Cao15}. In particular, the ``{\it XX} from {\it XZ} gadget'' of Biamonte and Love \cite{Biamonte08} can be employed to effect homophysically an $X\otimes X$ interaction between a first and a second rebits through $X\otimes Z$ interactions with a zeroth rebit,
\begin{align}
I - X_1\otimes X_2 \;\,\HomoPhysR\,\;& \gamma_0^2 \, (I - X_0) + (I - X_1\otimes X_2) \nonumber \\[0.75ex]
\,\HomoPhysR\,\;& \gamma_0^2 \, (I - X_0) + \gamma_{0\,} (X_1 + X_2) \otimes Z_0 + 2I + O(\gamma_0^{-1}) \,, \label{XXfromXZ}
\end{align}
where $\HomoPhysR$ reads and stands for ``is homophysically mapped to'', $\gamma_0 \gg 1$ is a large constant. Then the linear combinations of $X\otimes X$ and $X\otimes Z$ include all two-rebit interactions of the form $X\otimes R(\theta)$, with $R(\theta) \defeq Z\cos\theta + X\sin\theta$, $\theta \in [-\pi,\pi)$. Alternatively, there is a special class of multi-rebit interactions called {\it multi-rebit-controlled gates} of the form $R_i(\theta)\otimes\prod_{j\in J}Z_j^{\pm}$, with $\theta \in [-\pi,\pi)$, $i$ indexing a rebit being operated upon, $J$ being a set indexing a fixed number of control rebits. Such a multi-rebit-controlled $R$-gate does not require a decomposition into two-rebit couplings, but can be implemented through a linear combination of the following homophysics,
\begin{align}
\mathfrak{M}_{\,} \scalebox{1.15}{$($} X_i\otimes {\textstyle{ \prod_{j\,\in J\,} }} Z_j^{\pm} \scalebox{1.15}{$)$} \,=\,\;& (2\gamma_0^2/\pi^2) \,\scalebox{1.15}{$[$}\, {\sf H}_{\mathsmaller{\rm BF},i} +\, {\textstyle{ \sum_{j\,\in J\,} }} {\sf H}_{\mathsmaller{\rm BF},j} +\, {\sf X}(x_{i1},x_{i2}) \, {\textstyle{ \prod_{j\,\in J\,} }} {\sf Z}^{\pm}(x_{j1},x_{j2}) \,\scalebox{1.15}{$]$}\,, \label{MultiCtrlX} \\[0.75ex]
\mathfrak{M}_{\,} \scalebox{1.15}{$($} Z_i^{\pm}\otimes {\textstyle{ \prod_{j\,\in J\,} }} Z_j^{\pm} \scalebox{1.15}{$)$} \,=\,\;& \gamma_0 {\sf H}_{\mathsmaller{\rm BF},i} +\, \gamma_0 \, {\textstyle{ \sum_{j\,\in J\,} }} {\sf H}_{\mathsmaller{\rm BF},j} +\, {\sf Z}^{\pm}(x_{i1},x_{i2}) \, {\textstyle{ \prod_{j\,\in J\,} }} {\sf Z}^{\pm}(x_{j1},x_{j2}) \,. \label{MultiCtrlZ}
\end{align}
At any rate, it has been established that any computationally $k$-local Hamiltonian $H$ involving $n\in\mathbb{N}$ rebits, with $k\in\mathbb{N}$ being a fixed number and $n$ a variable, can be homophysically implemented as an SCFF Hamiltonian $\mathfrak{M}(H)$ involving no more than ${\rm poly}(n,\epsilon^{-1})$ bi-fermions, such that the low-energy physics of $H$ and $\mathfrak{M}(H)$ are homophysical up to an error tolerance $\epsilon > 0$, where each CFF interaction in $\mathfrak{M}(H)$ moves no more than $k'\in\mathbb{N}$ bi-fermions, with $k'$ being another fixed number, and has an operator norm that is upper-bounded by ${\rm poly}(\epsilon^{-1})$, while all bi-fermions are mutually distinguishable entities.

Given a universal BQP algorithm $\{U_t(\theta_t)\defeq R_{i(t)}(\theta_t)\mbox{\;or\;}R_{i(t)}(\theta_t)\otimes Z^+_{j(t)} + R_{i(t)}(-\theta_t)\otimes Z^-_{j(t)}\}_{t\in[1,\mathsmaller{T}]}$, $T\in\mathbb{N}$, with each free or controlled $R$-gate $U_t(\theta_t)$, $\theta_t\in[-\pi,\pi)$ operating on an $i(t)$-th, $i(t)\in[1,n]$ and possibly a $j(t)$-th, $j(t)\in[1,n]$ rebits in an $n$-rebit {\it logic register} represented by $({\cal C}_{\mathsmaller{L}} \defeq \{0,1\}^n,{\cal H}_{\mathsmaller{L}},{\cal B}_{\mathsmaller{L}})$ as a quantum subsystem, $\forall t\in[1,T]$, where the successive applications of the free or controlled $R$-gates are meant to generate a series of quantum states $|\phi_t\ra_{\mathsmaller{L}\!} \defeq U_t |\phi_{t{-}1}\ra_{\mathsmaller{L}}$, $t\in[1,T]$, from a given initial state $|\phi_0\ra_{\mathsmaller{L}}$ till a computational result $|\phi_{\mathsmaller{T}}\ra_{\mathsmaller{L}} = \scalebox{1.15}{(} \prod_{t=1}^{\mathsmaller{T}}U_t \scalebox{1.15}{)} |\phi_0\ra_{\mathsmaller{L}}$ at the end, the celebrated Feynman-Kitaev construct \cite{Feynman85,Kitaev02,Kempe06,Wei20} introduces a {\it clock register} represented by $({\cal C}_{\mathsmaller{C}}, {\cal H}_{\mathsmaller{C}}, {\cal B}_{\mathsmaller{C}})$ as a quantum subsystem to support {\it clock states} $\{|t\ra_{\mathsmaller{C}}\}_{t\in[0,\mathsmaller{T}]} \subseteq {\cal H}_{\mathsmaller{C}}$, so that the clock and logic registers constitute a GSQC system represented by $({\cal C}_{\mathsmaller{C}\!} \times {\cal C}_{\mathsmaller{L}}, {\cal H}_{\mathsmaller{C}} \otimes {\cal H}_{\mathsmaller{L}}, {\cal B}_{\mathsmaller{C}} \otimes {\cal B}_{\mathsmaller{L}})$, on which the product states $\{|t\ra_{\mathsmaller{C}} |\phi_t\ra_{\mathsmaller{L}}\}_{t\in[0,\mathsmaller{T}]} \subseteq {\cal H}_{\mathsmaller{C}} \otimes {\cal H}_{\mathsmaller{L}}$ map and encode the entire computational history of the BQP algorithm. Then Feynman's clocked Hamiltonians $H_{{\rm Feyn},\,t} \defeq |t\ra_{\mathsmaller{C}} \la(t{-}1)|_{\mathsmaller{C}} \otimes U_t + |(t{-}1)\ra_{\mathsmaller{C}} \la t|_{\mathsmaller{C}} \otimes U_t$, $t\in[1,T]$ ensure that the associated quantum gates $U_t$, $t\in[1,T]$ are applied to the logic register in the correct order when the clock register undergoes transitions between what he called the {\it program counter sites} (namely, the clock states, also referred to as {\it clock sites}) $|t\ra_{\mathsmaller{C}}$, $t\in[1,T]$ \cite{Feynman85}. Finally, Kitaev's GSQC Hamiltonian (also called the {\it Feynman-Kitaev Hamiltonian}) $H_{\mathsmaller{\rm FK}} \defeq H_{\rm clock} + H_{\rm init} + H_{\rm prop}$ enforces computational constraints via energy penalties, with $H_{\rm clock}$ restricting the clock register to the manifold of ${\rm span}(\{|t\ra_{\mathsmaller{C}}\! : t\in[0,T]\})$, $H_{\rm init}$ setting the initial state, while $H_{\rm prop}$ performing the quantum computation as Feynman suggested, such that the ground state $\psi_0(H_{\mathsmaller{\rm FK}}) = (T\,{+}\,1)^{-1/2} \sum_{t=0}^{\mathsmaller{T}} |t\ra_{\mathsmaller{C}}|\phi_t\ra_{\mathsmaller{L}}$ is unique and polynomially gapped \cite{Kitaev02,Kempe06}. It is WLOG to assume that the initial state $|\phi_0\ra_{\mathsmaller{L}}$ is a ${\cal C}_{\mathsmaller{L}}$-coordinate eigenstate, because, otherwise, it must be preparable from a ${\cal C}_{\mathsmaller{L}}$-coordinate eigenstate by another BQP algorithm. It is convenient and WLOG to assume that all logic rebits are initially set to their $|1\ra$ state \cite{Rudolph02}.

There are several choices of encoding a clock register for the clock states $\{|t\ra_{\mathsmaller{C}}\}_{t\in[0,\mathsmaller{T}]}$ \cite{Nagaj10,Breuckmann14}. Take Kitaev's {\it domain wall clock} for example, which can be realized by a clock register consisting of $T\,{+}\,2$ rebits indexed by integers within $[0,T\,{+}\,1]$ and using $|t\ra_{\mathsmaller{C}} \defeq |1\ra_{\mathsmaller{C}}^{\otimes(t+1)}|0\ra^{\otimes(\mathsmaller{T}-t+1)}_{\mathsmaller{C}}$, $t \in [0,T]$, such that \cite{Kitaev02,Kempe06,Wei20}
\begin{align}
H_{\mathsmaller{\rm FK}} \,\defeq\;& \gamma_{0\,} H_{\rm clock\,} +\, \gamma_{0\,} H_{\rm init\,} +\, \scalebox{1.2}{$\sum$}_{t=1}^{\mathsmaller{T}} \, H_{{\rm prop},\,t} \,, \label{HamilFK}
\iftoggle{UseREVTEX}{\\}{\\[0.75ex]}
H_{\rm clock} \,\defeq\;& Z^+_{\mathsmaller{C},\,0} \,+\, Z^-_{\mathsmaller{C},\,\mathsmaller{T+1}\,} +\, \scalebox{1.2}{$\sum$}_{t=1}^{\mathsmaller{T}} \, Z^+_{\mathsmaller{C},\,t-1} \otimes Z^-_{\mathsmaller{C},\,t} \,, \label{HamilClock}
\iftoggle{UseREVTEX}{\\}{\\[0.75ex]}
H_{\rm init} \,\defeq\;& \scalebox{1.2}{$\sum$}_{i=1}^n \, Z^+_{\mathsmaller{C},\,1} \otimes Z^+_{\mathsmaller{L},\,i} \,, \label{HamilInit}
\iftoggle{UseREVTEX}{\\}{\\[0.75ex]}
H_{{\rm prop},\,t} \,\defeq\;& Z^-_{\mathsmaller{C},\,t-1} \otimes Z^+_{\mathsmaller{C},\,t+1} \otimes [I - X_{\mathsmaller{C},\,t} \otimes U_t(\theta_t)]\,, \; \forall t \in [1,T] \,, \label{HamilProp}
\end{align}
where $\theta_t\in[-\pi,\pi)$ for all $t\in[1,T]$, $Q^{\delta}_{\mathsmaller{C},\,t}$ means to apply a single-rebit operator $Q^{\delta}$ to the $t$-th rebit of the clock register, $\forall Q \in \{X,Z\}$, $\forall \delta \in \{+,-,\mbox{void}\}$, $\forall t \in [0,T\,{+}\,1]$, the energy constant $\gamma_0>0$ is sufficiently large but still $O({\rm poly}(T\,{+}\,n))$-bounded such that the system can only move in the ground state subspace of $H_{\rm clock}+H_{\rm init}$, any escape from $\psi_0(H_{\rm clock}+H_{\rm init})$ is exponentially suppressed and negligible. It is clear that ${\rm size}(H_{\mathsmaller{\rm FK}}) = O(T+n)$. With each of $Z^{\pm}_{\mathsmaller{C},\,t}$, $t\in[0,T+1]$, $Z^+_{\mathsmaller{L},\,i}$, $i\in[1,n]$, and $h_{{\rm prop},\,t} \defeq I - X_{\mathsmaller{C},\,t} \otimes U_t$, $t\in[1,T]$ regarded as a single operator, the Feynman-Kitaev Hamiltonian $H_{\mathsmaller{\rm FK}}$ is a sum of {\it few-body moving (FBM) tensor monomials} \cite{Wei20} as specified in equations (\ref{HamilClock}-\ref{HamilProp}), each of which is a tensor product of no more than three single operators that involves no more than five interacting rebits in total. For each $t \in [1,T]$, the single operator $h_{{\rm prop},\,t}$ and the FBM tensor monomial $H_{{\rm prop},\,t} = Z^-_{\mathsmaller{C},\,t-1} \otimes Z^+_{\mathsmaller{C},\,t+1} \otimes h_{{\rm prop},\,t}$ are called the $t$-th {\it free and controlled Feynman-Kitaev propagators} respectively.

It is straightforward to implement such an $H_{\mathsmaller{\rm FK}}$ into an $\mathfrak{M}(H_{\mathsmaller{\rm FK}})$ for a system of $2T\,{+}\,n\,{+}\,2$ bi-fermions, where each of the $T\,{+}\,2$ clock rebits and $n$ logic rebits corresponds to one unique bi-fermion, each of the FBM tensor monomials in equations (\ref{HamilClock}) and (\ref{HamilInit}) is mapped to an interaction among the corresponding bi-fermions as in equation (\ref{MultiCtrlZ}), while the remaining $T$ bi-fermions supply enough auxiliary rebits for perturbative gadgets to implement logic gates of the form $X\otimes R(\theta)$, $\theta\in[-\pi,\pi)$ for Feynman-Kitaev propagators. Such an $\mathfrak{M}(H_{\mathsmaller{\rm FK}}) = \sum_{k=1}^{\mathsmaller{K}} \mathfrak{M}(H_k)$, $K \defeq 2T+n+2$ as an SCFF Hamiltonian is both LTK-decomposed and GSP-decomposed, with each FBM tensor monomial $H_k$, $k\in[1,K]$ being either of the form $Z_1^{\pm}\otimes Z_2^{\pm}$ or a tensor product of an FBM tensor monomial $Z_1^{-\!}\otimes Z_2^+$ and a free Feynman-Kitaev propagator $h_{{\rm prop},\,t}$, $t\in[1,T]$, where $Z_1$ and $Z_2$ represent a $Z$-operator acting on a clock or logic rebit. The unique ground state of $\mathfrak{M}(H_{\mathsmaller{\rm FK}})$ can be simulated using Monte Carlo on a classical computer by repeating a sequence of Gibbs operators $\{\exp[-\tau\mathfrak{M}(H_k)]\}_{k=1}^{\mathsmaller{K}}$ with a suitable $\tau>0$ such that $\tau + \tau^{-1} = O({\rm poly}(T+n))$ for an $O({\rm poly}(T+n))$-bounded number of times to approximates a Gibbs operator associated with $\mathfrak{M}(H_{\mathsmaller{\rm FK}})$. It is also straightforward to construct an adiabatic sequence of Feynman-Kitaev Hamiltonians $\{H_{\mathsmaller{\rm FK}}(l) : l\in[0,m]\}$ that evolves from an initial $H_{\mathsmaller{\rm FK}}(0)$ with a known ground state $\psi_0(H_{\mathsmaller{\rm FK}}(0))$ to the final $H_{\mathsmaller{\rm FK}}(0) = H_{\mathsmaller{\rm FK}}$ of equation (\ref{HamilFK}). An exemplary embodiment uses an adiabatic sequence of Feynman-Kitaev Hamiltonians with the $l$-th Hamiltonian $H_{\mathsmaller{\rm FK}}(l) \defeq \gamma_{0\,} H_{\rm clock} + \gamma_{0\,} H_{\rm init} + \sum_{t=1}^{\mathsmaller{T}} H_{{\rm prop},\,t\,}(l)$, $H_{{\rm prop},\,t}(l) \defeq Z^-_{\mathsmaller{C},\,t-1} \otimes Z^+_{\mathsmaller{C},\,t+1} \otimes h_{{\rm prop},\,t\,}(l)$, and $h_{{\rm prop},\,t\,}(l) \defeq I - X_{\mathsmaller{C},\,t} \otimes U_t(l\theta_t/m)$, $\forall l\in[0,m]$, in conjunction with an initial ground state $\psi_0(H_{\mathsmaller{\rm FK}}(0)) \defeq (T\,{+}\,1)^{-1/2} \sum_{t=0}^{\mathsmaller{T}} |t\ra_{\mathsmaller{C}}|\phi_0\ra_{\mathsmaller{L}}$.

For a sequence of Gibbs operators $\{\exp[-\tau\mathfrak{M}(H_k)]\}_{k=1}^{\mathsmaller{K}}$, $\tau>0$ associated with a sequence of FBM tensor monomials $\{H_k\}_{k=1}^{\mathsmaller{K}}$ LTK- and GSP-decomposing a Feynman-Kitaev Hamiltonian $H_{\mathsmaller{\rm FK}} = \sum_{k=1}^{\mathsmaller{K}} H_k$, when $H_k$, $k\in[1,K]$ is $({\cal C}_{\mathsmaller{C}\!} \times {\cal C}_{\mathsmaller{L}})$-diagonal of the form $Z_1^{\pm}\otimes Z_2^{\pm}$, it is obvious that $\exp[-\tau\mathfrak{M}(H_k)]$ has the same amplitude integral $D_k(\pi_kq_{k-1})$ with respect to any $\pi_k\in A_*$ and any $q_{k-1}$ in the support of any ground state of $\mathfrak{M}(Z_1^{\pm}\otimes Z_2^{\pm})$, while the probability of encountering a $q_{k-1}$ out of the supports of the ground states of $\mathfrak{M}(Z_1^{\pm}\otimes Z_2^{\pm})$ can be made negligible by choosing a sufficiently large energy constant $\gamma_0$. For a Gibbs operator of the form $\exp[-\tau\mathfrak{M}(X\otimes R(2\theta))]$, $\theta\in[-\pi/2,\pi/2)$, the only relevant states are the two ground states $|+\ra|\theta_+\ra$, $|-\ra|\theta_-\ra$ and the two lowest excited states $|+\ra|\theta_-\ra$, $|-\ra|\theta_+\ra$, with $|\pm\ra \defeq (|0\ra \pm |1\ra)/\sqrt{2}$, $|\theta_+\ra \defeq \cos\theta|0\ra + \sin\theta|1\ra$, $|\theta_-\ra \defeq \cos\theta|1\ra - \sin\theta|0\ra$. It is easy to verify that for any of the cases of $q_{k-1}$ falling into the support of the basis state $\mathfrak{M}(|0\ra|0\ra)$, $\mathfrak{M}(|0\ra|1\ra)$, $\mathfrak{M}(|1\ra|0\ra)$, or $\mathfrak{M}(|1\ra|1\ra)$, the amplitude integral $D_k(q_{k-1})$ of $\la\cdot|\exp[-\tau\mathfrak{M}(X\otimes R(2\theta))]|q_{k-1}\ra$ always yields the same value $(1+e^{-\tau}) + (1-e^{-\tau})(|\cos 2\theta|+|\sin 2\theta|)$, so long as the strength of the Dirac potential barrier and well for each bi-fermion is set to $\gamma_0>0$, which is sufficiently large such that the wavefunction within each logic well of each bi-fermion always approximates a half-sine with an $O(\gamma_0^{-2})$ error \cite{Wei20}. Therefore, any Gibbs operator $\exp[-\tau\mathfrak{M}(X\otimes R(2\theta))]$ associated with a free Feynman-Kitaev propagator involving a free $R$-gate $R(2\theta)$ does not induce any path dependency of amplitude integrals. It follows straightforwardly that the same is true for a Gibbs operator associated with a free Feynman-Kitaev propagator involving a controlled $R$-gate $R(2\theta)\otimes Z^{+\!} + R(-2\theta)\otimes Z^{-\!}$ because the amplitude integral yields the same value regardless of the control logic rebit is in the null state of $Z^{+\!}$ or $Z^-$. Finally, any Gibbs operator $\exp[-\tau\mathfrak{M}(Z^{-\!}\otimes Z^{+\!}\otimes(I-X\otimes U))]$ associated with a controlled Feynman-Kitaev propagator with any $\tau = O({\rm poly}(T+n))$ can be realized by repeating applications of $G_{\rm free}(\delta\tau) \defeq \exp[-\delta\tau\mathfrak{M}(I-X\otimes U)]$ followed by $G_{\rm ctrl}(\delta\tau) \defeq \exp[-\delta\tau\mathfrak{M}(Z^{-\!}\otimes Z^{+\!}\otimes(I-X\otimes U)+(I-Z^{-\!}\otimes Z^{+\!})\otimes(I+X\otimes U))]$ for $m = O({\rm poly}(T+n))$ times, $\delta\tau \defeq \tau/m$, such that each of $G_{\rm free}(\delta\tau)$ and $G_{\rm ctrol}(\delta\tau)$ is already AIB, and the product operator $\exp[-\tau\mathfrak{M}(Z^{-\!}\otimes Z^{+\!}\otimes(I-X\otimes U))] = [G_{\rm ctrl}(\delta\tau){\kern0.1em}G_{\rm free}(\delta\tau)]^m + O(1/{\rm poly}(m))$ is naturally AIB.

\begin{theorem}{} \label{TheoremTwo}
A homophysics $\mathfrak{M}$ exists, which maps the Feynman-Kitaev Hamiltonian $H_{\mathsmaller{\rm FK}}$ as defined in equations (\ref{HamilFK}-\ref{HamilProp}) to an SCFF Hamiltonian $\mathfrak{M}(H_{\mathsmaller{\rm FK}})$ that is both LTK- and GSP-decomposed into a sequence of CFF interactions, such that any Gibbs operator generated by $\mathfrak{M}(H_{\mathsmaller{\rm FK}})$ is well approximated by an AIB sequence of Gibbs operators generated by said sequence of CFF interactions.
\iftoggle{UseREVTEX}{}{\vspace{-1.5ex}}
\end{theorem}

\begin{proof}[Proof]
An $\mathfrak{M}(H_{\mathsmaller{\rm FK}})$ as constructed above is obviously SCFF and frustrate-free \cite{Bravyi10,Jordan10,Wei20}, which is both LTK- and GSP-decomposed into a sequence of CFF interactions, where each CFF interaction homophysically implements one of the FBM tensor monomials in equations (\ref{HamilClock}-\ref{HamilProp}), involves no more than $6$ bi-fermions and moves even less. For any $\tau>0$, $\tau = O({\rm poly}(T+n))$, an AIB sequence of an $O({\rm poly}(T+n))$ number of Gibbs operators can be generated from the sequence of CFF interactions as described above, whose product approximates $\exp[-\tau\mathfrak{M}(H_{\mathsmaller{\rm FK}})]$ to within an $O(1/{\rm poly}(T+n))$ accuracy. 
\iftoggle{UseREVTEX}{}{\vspace{-1.5ex}}
\end{proof}

It is noted that the above-specified technique to render a sequence of Gibbs operators AIB is only by way of example but no means of limitation. It should be obvious to one skilled in the art that many other ways can achieve the same effect. Furthermore, the AIB property is only used for mathematical rigor in theory to establish rapid mixing of one particular Monte Carlo sampling method. A practical Monte Carlo simulation can sample restricted Feynman paths efficiently using an alternative method without requiring an AIB property of the associated sequence of Gibbs operators.

\begin{theorem}{} \label{TheoremThree}
BQP $\subseteq$ BPP, therefore BPP $=$ BQP, as BPP $\subseteq$ BQP is well known.
\iftoggle{UseREVTEX}{}{\vspace{-1.5ex}}
\end{theorem}

\begin{proof}[Proof]
Using the Feynman-Kitaev construct as specified in equations (\ref{HamilFK}-\ref{HamilProp}), any BQP algorithm of $T\in\mathbb{N}$ gates on a quantum register of $n\in\mathbb{N}$ rebits can be mapped to the ground state of a Feynman-Kitaev Hamiltonian $H_{\mathsmaller{\rm FK}}$ of an $O({\rm poly}(T+n))$ size, whose ground state is non-degenerate and polynomial-gapped. Moreover, an adiabatic sequence of Feynman-Kitaev Hamiltonians $\{H_{\mathsmaller{\rm FK}}(l)\}_{l=0}^{m}$, $m\in\mathbb{N}$ can be designed to reach the final $H_{\mathsmaller{\rm FK}}(m)=H_{\mathsmaller{\rm FK}}$ from an initial $H_{\mathsmaller{\rm FK}}(0)$ with trivial Feynman-Kitaev propagators $h_{{\rm prop},\,t\,}(0) \defeq I - X_{\mathsmaller{C},\,t}$, $\forall t\in[1,T]$ and a trivial ground state $\psi_0(H_{\mathsmaller{\rm FK}}(0)) \defeq (T\,{+}\,1)^{-1/2} \sum_{t=0}^{\mathsmaller{T}} |t\ra_{\mathsmaller{C}}|\phi_0\ra_{\mathsmaller{L}}$.

By \myTheorem~\ref{TheoremTwo}, each $H_{\mathsmaller{\rm FK}}(l)$, $l\in[0,m]$ can be homophysically mapped to an SCFF Hamiltonian $\mathfrak{M}[H_{\mathsmaller{\rm FK}}(l)]$ that is both LTK- and GSP-decomposed into a sequence of CFF interactions, which generates an AIB sequence of Gibbs operators to project out the ground state $\psi_0(H_{\mathsmaller{\rm FK}}(l))$ to within an $O(1/{\rm poly}(T+n))$ accuracy. \myTheorem~\ref{TheoremOne} guarantees an FPRAS, which is in BPP, to simulate an AIB sequence of Gibbs operators comprising concatenated subsequences of Gibbs operators, each subsequence corresponding to an $H_{\mathsmaller{\rm FK}}(l)$, $l\in[0,m]$. This establishes BQP $\subseteq$ BPP, therefore BPP $=$ BQP.
\iftoggle{UseREVTEX}{}{\vspace{-1.5ex}}
\end{proof}

In conclusion, it has been proved that BPP and BQP are exactly the same computational complexity class. As a consequence, any computational problem admitting a BQP algorithm has a BPP solution by mapping the BQP algorithm into a GSQC problem, and simulating the GSQC system by Monte Carlo on a BPP machine. The overall method is called {\it Monte Carlo quantum computing} (MCQC) \cite{Wei20}. The significance of BPP$\,=\,$BQP can hardly be overstated. Despite providing a negative answer to the long outstanding question of whether the laws of quantum mechanics endow more computational power, it opens up new avenues for developing and identifying efficient algorithms for classical computers from the vantage point of quantum computing. Any quantum based or inspired solution to a computational problem translates automatically into an efficient classical probabilistic algorithm. For instance, it is now certain that integer factorization is in BPP, thanks to Shor's celebrated quantum discovery \cite{Shor97} that catalyzed the quest of quantum computing. Besides the standard decision or promise problems \cite{Even84,Goldreich06,Arora09} in BPP$\,=\,$BQP, there are a vast number of problems in the forms of function evaluation, objective optimization, and target search, {\it etc.} \cite{Arora09} that are either equivalent or polynomially reducible to a problem in BQP, therefore, also efficiently solvable via MCQC. An excellent example is the quantum algorithm of Harrow, Hassidim, and Lloyd for linear systems of equations \cite{Harrow09}, which now gets an efficient MCQC implementation.

Being able to simulate quantum systems efficiently was one of the reasons that Feynman and others proposed quantum computing and quantum computers \cite{Feynman85,Lloyd96}. One important application of MCQC without a sign problem is naturally to simulate many-body quantum systems efficiently via Monte Carlo on a classical computer, by just constructing a BQP algorithm simulating the quantum system, then mapping the BQP algorithm to an efficient MCMC through a Feynman-Kitaev construct. Such many-body quantum systems include any FSFS as a special case.

On the other hand, an FSFS can be directly simulated via a restricted path integral method without devising a BQP algorithm then using a Feynman-Kitaev construct. In one exemplary embodiment for an FSFS of $S\in\mathbb{N}$ species, each species $s\in[1,S]$ having $n_s\in\mathbb{N}$ identical fermions, the Hamiltonian is a Schr\"odinger operator and written as $H = \sum_{s=1}^{\mathsmaller{S}} \sum_{l=1}^{n_s} \sum_{m=1}^{n_s} H_{slm}$, with $H_{slm} = {-}(\Delta_{sl}+\Delta_{sm})/2n_s + V/\sum_{s=1}^{\mathsmaller{S}}n_s^2$ moves only the $l$-th and the $m$-th labeled particle of the $s$-th species through the Laplace-Beltrami operators $\Delta_{sl}$ and $\Delta_{sm}$, for each $(s,l,m)\in[1,S]\times[1,n_s]^2$. Using the method of LTK-decomposition, a Gibbs operator $e^{-\tau H}$, $\tau\in(0,\infty)$ is approximated by applying a sequence of Gibbs operators $\{{\sf G}_{slm} \defeq e^{-\tau H_{slm}/m}\}_{(s,l,m)}$ for a polynomial-bounded $m\in\mathbb{N}$ times, each ${\sf G}_{slm}$ associated with a single Feynman slice being partially fermionic exchange-symmetrized with respect to an order-$2$ group $G_{slm}$ that permutes the $l$-th and $m$-th identical fermions of the $s$-th species, for all $(s,l,m)\in[1,S]\times[1,n_s]^2\}$. By a theorem of Diaconis and Shahshahani \cite{Diaconis81}, a typical Feynman path with such partial fermionic symmetrizations effectively selects a sample of fully symmetrized and signed Feynman path almost surely and uniformly, such that an integral of signed Wiener measure densities of said Feynman paths with partial symmetrizations well approximate the FSFS. It is important to note that $\forall(s,l,m)\in[1,S]\times[1,n_s]^2$, $\forall\tau\in(0,\infty)$, and any given configuration point $q$, the partially fermionic symmetrized Gibbs kernels $\la\cdot|e^{-\tau H_{slm}}|{\cal F}_{slm\,}q\ra = \la{\cal F}_{slm\!}\cdot|e^{-\tau H_{slm}}|q\ra = \la{\cal F}_{slm\!}\cdot|e^{-\tau H_{slm}}|{\cal F}_{slm\,}q\ra$, with ${\cal F}_{slm} \defeq |G_{slm}|^{-1} \sum_{\,\pi\in G_{slm}} (-1)^{\pi\,} \pi$, are all efficiently computable whose nodal surfaces can be determined at an $O({\rm poly}({\rm size}(H),\epsilon))$ computational cost, $\epsilon$ being any desired numerical accuracy.

Changing the perspective again, when an MFS comprises a certain fermion species that has many identical fermions residing separately in multiple non-overlapping regions of a {\it substrate space} \cite{Wei20} or a phase space ({\it e.g.}, a position-momentum space) of single particles, then the identical fermions are divided into separated clusters each of which corresponds to a specific spatial region separated from other spatial regions and becomes a unique effective species distinguishable from other effective species corresponding to other spatial regions \cite{Messiah99}, where the identical fermions within each effective species residing in a separated spatial region are indistinguishable and obey fermionic exchange symmetry, whereas different clusters of fermions residing in different spatial regions are mutually distinguishable as different effective species. Such an MFS is homophysical to an effective system comprising multiple effective species.

Finally, it is useful to note that the analyses, algorithms, and methods presented {\it supra} can be extended straightforwardly to physical and computational systems over a discrete or continuous-discrete product configuration space \cite{Wei20}, only that the nodal restriction or path rectification may need to invoke the so-called {\it lever rule} \cite{vanBemmel94,Sorella15,Wei20} using efficiently solved nodal structures of CFF interactions or their associated Gibbs kernels. Besides, when a compact configuration space is approximated by a lattice with the spacing between neighbor lattice points being sufficiently small, the error in localizing nodal surfaces becomes negligibly small and the simulation results from continuous and discrete configuration spaces converge.

\iftoggle{UseREVTEX}{}{\vspace{-1.5ex}}

\end{document}